\documentclass[a4paper, authoryear, 11pt]{elsarticle}

\usepackage[english]{babel}
\usepackage{babel}
\usepackage{fontenc}
\usepackage{graphicx}
\usepackage{amsmath,amsthm,amssymb}
\usepackage{natbib}
\usepackage[colorlinks=true,citecolor=blue]{hyperref}
\usepackage{color}
\usepackage{graphicx}
\usepackage{subfigure}
\usepackage{color}
\usepackage{newlfont}
\usepackage{multirow}
\usepackage{longtable} 
\usepackage{ifthen}
\usepackage{alltt}
\usepackage{enumerate}
\usepackage[text={6.25in,8.5in},centering]{geometry}
\usepackage{tikz}
\usepackage{epstopdf}
\usepackage{float}
\usepackage{arydshln}
\usepackage[utf8]{inputenc}
\usepackage{tikz}
\usetikzlibrary{shapes,arrows}
\numberwithin{equation}{section}
\newtheorem{theorem}{Theorem}[section]

\newtheorem{proposition}{\bf Proposition}[section]

\newtheorem{corollary}{\hskip\parindent\bf Corollary}[section]

\input epsf.sty
\usepackage{lineno}

\begin{document}
	\title{Metapopulation dynamics of a respiratory disease with infection during travel}
	\author[IISc]{Indrajit Ghosh \footnote{Corresponding author. Email: indra7math@gmail.com}}
	\author[ISI]{Sk Shahid Nadim}
	\author[IISc]{Soumyendu Raha}
	\author[IISc]{Debnath Pal}
	\address[IISc]{Department of Computational and Data Sciences, Indian Institute of Science, Bengaluru-560012, Karnataka, India}
    \address[ISI]{Agricultural and Ecological Research Unit, Indian Statistical Institute, Kolkata - 700 108, West Bengal, India}

\begin{abstract}
	We formulate a compartmental model for the propagation of a respiratory disease in a patchy environment. The patches are connected through the mobility of individuals, and we assume that disease transmission and recovery are possible during travel. Moreover, the migration terms are assumed to depend on the distance between patches and the perceived severity of the disease. The positivity and boundedness of the model solutions are discussed. We analytically show the existence and global asymptotic stability of the disease-free equilibrium. We study three different network topologies numerically and find that underlying network structure is crucial for disease transmission. Further numerical simulations reveal that infection during travel has the potential to change the stability of disease-free equilibrium from stable to unstable. The coupling strength and transmission coefficients are also very crucial in disease propagation. Different exit screening scenarios indicate that the patch with the highest prevalence may have adverse effects but other patches will be benefited from exit screening. Furthermore, while studying the multi-strain dynamics, it is observed that two co-circulating strains will not persist simultaneously in the community but only one of the strains may persist in the long run. Transmission coefficients corresponding to the second strain are very crucial and show threshold like behavior with respect to the equilibrium density of the second strain.

\end{abstract}


\maketitle

\section{Introduction}\label{section1}
Many epidemic outbreaks such as 1918 pandemic influenza, 2002-2003 SARS outbreak, 2009 H1N1 influenza epidemic, 2012 - 2015 MERS-CoV outbreak, 2015 Ebola virus epidemic and ongoing COVID-19 pandemic indicate that increasing connectivity has significantly amplified the impact of these diseases. Respiratory diseases such as TB, influenza, COVID-19, etc. are mainly caused by direct contact, indirect contact, respiratory droplets or airborne transmission. Direct contact refers to the human-to-human physical contact and indirect contact indicate contact through intermediate object such as door knob, sitting bench etc. Respiratory droplets are exhaled from infectious humans while sneezing or coughing. These droplets can then be deposited on a healthy human's mucus or conjunctiva. Another transmission path is the airborne one which involves infectious pathogens travelling through air from an infected host to a healthy one \cite{brankston2007transmission,arino2016revisiting}. However, pathogens may not only transmit within patches but also during travel. As during travel, individuals are put in a close proximity for a significant duration, it is highly likely that the probability of transmission increases with duration of transport or equivalently distance travelled. Therefore the infection during travel may be negligible in short-range transport but in the case of long-range transport infection during travel may have a significant impact. Transmission of respiratory diseases during travel has been reported in many instances. In a report by the European Centre for Disease prevention and Control stated that TB, measles and seasonal influenza are transmissible during commercial flights \cite{knipl2016stability}. WHO confirmed influenza transmission during long distance train travels \cite{furuya2007risk}. Mangili and Gendreau  \cite{mangili2005transmission} reviewed about transmission while travelling in aircraft, and infection spread within cars has also been studied \cite{knibbs2012risk}. Consequently, the infection during travel is important to study. The specific impact of infection while traveling is not well understood. Specifically, in the case of pandemic outbreaks (for instance influenza, COVID-19 etc.) the virus is highly infectious and may intensify the overall disease burden if long-range travels are not supervised. 

Mathematical modeling can shed light on the potential impact of infection during transportation on the global burden of respiratory illness. To analyze the transmission patterns of these diseases without infection during travel, various metapopulation models have been studied in the literature (see \cite{arino2017spatio, lloyd2004spatiotemporal,bahl2011temporally} and references therein). Despite its importance, the consideration of infection during transport in these metapopulation models is not yet emerged as a very active research area. Previously, various authors considered single patch and inflow of infected individuals. A disease transmission model in a single patch with inflow of infectious agents was considered \cite{brauer2001models}. Guo and coauthors investigated problems associated with the inflow of people infected with TB \cite{guo2011global}. The problem has been discussed by some authors in the context of delay differential equations to incorporate precise travel time between patches \cite{liu2008modeling,nakata2011global,knipl2016stability}. A few basic SIS (Susceptible-Infected-Susceptible) models (ordinary differential equation) are studied with infection during travel \cite{takeuchi2006spreading,takeuchi2007global,arino2016revisiting}. Takeuchi and coauthors \cite{takeuchi2006spreading,takeuchi2007global} presented SIS patch models with infection during transport. But they made a simplifying assumption that all model parameters are the same in both patches. However, Arino and coauthors \cite{arino2016revisiting} modified an earlier model of Takeuchi and coauthors \cite{takeuchi2006spreading} to allow all the parameters to be different in two patches.
However, SIS-type models are inadequate for respiratory diseases with high infectiousness (eg. influenza, SARS, measles and COVID-19). For instance, respiratory diseases such as measles \cite{bolker1993chaos}, COVID-19 \cite{wu2020nowcasting}, tuberculosis \cite{liu2010tuberculosis}, influenza A \cite{etbaigha2018seir} were modelled using an SEIR (Susceptible-Exposed-Infected-Recovered) system. This indicate that SEIR type models are necessary to model these respiratory diseases. Therefore, we extend the metapopulation models to SEIR version theoretically. Due to the consideration of SEIR system, travel related recovery and incubation period are also important to model in this context. Another realistic modification will consider the migration rates depending on the distance between patches and perceived severity of the disease. To do so, we may assume the migration rates to be inversely proportional to the distance between patches and inversely proportional to the perceived severity of the disease. These two assumptions are realistic but they were not considered in the context of infection during travel.
Moreover, respiratory diseases are prone to mutation and multiple strains often emerge in the population \cite{lyons2018mutation,garcia2021multiple}. Thus, it is important to incorporate multi-strain dynamics in the model under investigation. Motivated by the above discussion, we study the dynamics of a metapopulation model with infection during travel in different scenarios. 

The rest of the paper is organized as follows: in Sect. \ref{section2}, the assumptions for the metapopulation model are stated and the corresponding model is formulated; in Sect. \ref{section3}, some basic mathematical properties of the model are deduced; in Sect. \ref{section4}, different numerical simulations are performed to get insight into the transmission dynamics and finally in Sect. \ref{section5}, a brief discussion about the results is presented.

\section{Model description}\label{section2}
The following assumptions are made while formulating the model.
\begin{itemize}
    \item {We consider n disjoint patches with total populations $N_i$, i=1,2,...,n. The total population in each patch are subdivided into four compartments namely, Susceptible ($S$), Exposed ($E$), Infected ($I$) and Recovered ($R$).}
    \item {There is homogeneous mixing between people. Within each patch standard incidence rate of disease transmission is occurring, i.e, force of infection = $\frac{\beta_i S_i I_i}{N_i}$.}
    \item {Newly recruited people in each patch are susceptible in their respective patch. Individuals in each patch die at natural death rate. Disease induced death rate is considered only in compartments $I_i$. Exposed people will progress to infectious compartment after an incubation period in each patch.}
    \item {During travel new infection take place at rate $\frac{\alpha_j m_{ij} S_j I_j}{N_j}$, $m_{ij}$ is the rate of travel of each compartment from patch j to patch i \cite{arino2016revisiting}.}
    \item {While travelling, exposed people may become infectious and infected people will also recover at certain rates.}
    \item {The migration rates between patches are assumed to be inversely proportional to the distance between patches as well as the perceived severity of the disease. Thus, $m_{ij}=\epsilon \frac{1}{d_{ij}^{\eta}  \nu_j}$, where $d_{ij}$ is the distance between patch i and patch j, $\eta$ governs the distance dependency of migration rate, $\nu_j \geq 1$ measure the perceived severity of the disease. Note that, when $\nu_j = 1$ the migration rates do not depend on the disease severity.}
\end{itemize}

The SEIR metapopulation model based on the above assumptions take take the following form

\begin{eqnarray}\label{EQ:eqn 2.1}
\displaystyle{\frac{dS_i}{dt}} &=& \Pi_i- \beta_i \frac{I_i}{N_i} S_i - \mu_i S_i - \sum_{j=1}^{n} m_{ji} S_i + \sum_{j=1}^{n} m_{ij} (1-\frac{\alpha_j I_j}{N_j}) S_j,\nonumber \\
\displaystyle{\frac{dE_i}{dt}} &=& \beta_i \frac{I_i}{N_i} S_i -(\gamma_i+\mu_i)E_i - \sum_{j=1}^{n} m_{ji} E_i + \sum_{j=1}^{n} m_{ij} \frac{\alpha_j I_j}{N_j} S_j + \sum_{j=1}^{n} (1-\xi_i)m_{ij} E_j, \\
\displaystyle{\frac{dI_i}{dt}} &=& \gamma_i E_i - (\sigma_i+ \mu_i +\delta_i)I_i - \sum_{j=1}^{n} m_{ji} I_i + \sum_{j=1}^{n} \xi_i m_{ij} E_j + \sum_{j=1}^{n} (1-p_i) m_{ij} I_j, \nonumber \\
\displaystyle{\frac{dR_i}{dt}} &=& \sigma_i I_i - \mu_i  R_i - \sum_{j=1}^{n} m_{ji} R_i + \sum_{j=1}^{n} m_{ij} R_j + \sum_{j=1}^{n} p_i m_{ij} I_j, \nonumber
\end{eqnarray}

where, $m_{ii}=m_{jj}=0$ and $m_{ij}=\epsilon \frac{1}{d_{ij}^{\eta} \nu_j}$. A flow diagram of a source-destination patch version of model \eqref{EQ:eqn 2.1} is illustrated in Fig. \ref{fig:flow_diagram}.

\begin{figure}
    \centering
    \includegraphics[width=0.75\textwidth]{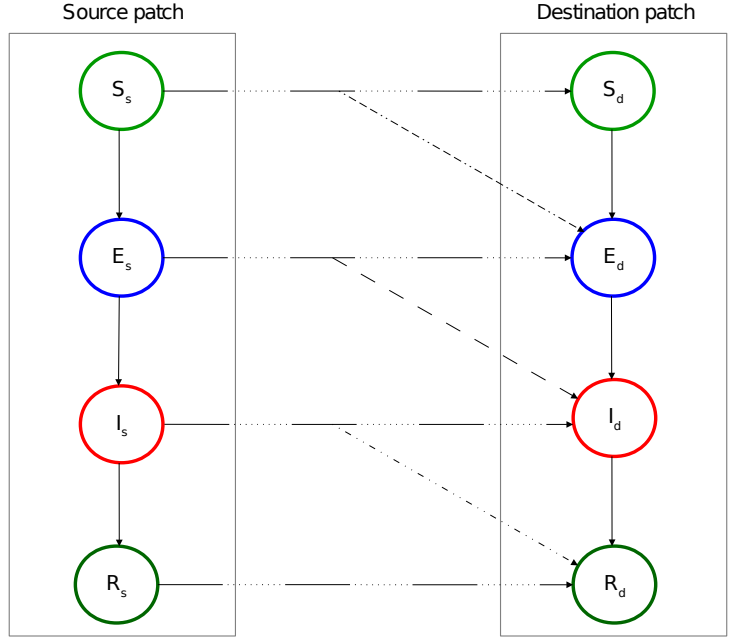}
    \caption{Flow chart of a source-destination version of SEIR metapopulation model with infection during travel. Solid black arrows depict usual progressions of a compartmental SEIR model and black arrows with fine dots are used to show migration of people. The black dashed arrow (2 dots 3 dashes) represents newly exposed persons during travel, the black dashed line represents the persons becoming infected from exposed and the black dashed line (2 dots one dash) represents the recovery during travel.}
    \label{fig:flow_diagram}
\end{figure}

\section{Mathematical analysis}{\label{section3}}
\subsection{Positivity and boundedness of the solution}
This subsection is provided to prove the positivity and boundedness of solutions of the system \eqref{EQ:eqn 2.1} with initial conditions $(S_i(0),E_i(0),I_i(0),R_i(0))^T\in \mathbb{R}_{+}^{4n}$.

\begin{proposition}
	The system \eqref{EQ:eqn 2.1} is positively invariant in $\mathbb{R}_{+}^{4n}$.
\end{proposition}
\begin{proof}
    To demonstrate the solution's positivity, it is sufficient to prove that each of the positive orthant's faces cannot be crossed, implying that the vector field points inward on the boundary of $\mathbb{R}_{+}^{4n}$.\\
	By re-writing the system \eqref{EQ:eqn 2.1} we have
	\begin{eqnarray}
	\frac{dX_i}{dt} &= F(X_i(t), X_0), X_{i0}\geq 0
	\label{EQ:eqn 2.2}
	\end{eqnarray}
	$ F(X_i(t))=(F_1(X_i),F_2(X_i),F_3(X_i),F_4(X_i))^T$\\
	We note that
	\begin{align*}
	\frac{dS_i}{dt}|_{S_i=0}&=\Pi_i+ \sum_{j=1}^{n} m_{ij} (1-\frac{\alpha_j I_j}{N_j})\geq 0,\\
	\frac{dE_i}{dt}|_{E_i=0}&= \beta_i \frac{I_i}{N_i} S_i + \sum_{j=1}^{n} m_{ij} \frac{\alpha_j I_j}{N_j} S_j + \sum_{j=1}^{n} (1-\xi_i)m_{ij} E_j \geq 0,\\
	\frac{dI_i}{dt}|_{I_i=0}&=\gamma_i E_i  + \sum_{j=1}^{n} \xi_i m_{ij} E_j + \sum_{j=1}^{n} (1-p_i) m_{ij} I_j \geq 0,\\
	\frac{dR_i}{dt}|_{R_i=0}&=\sigma_i I_i  + \sum_{j=1}^{n} m_{ij} R_j + \sum_{j=1}^{n} p_i m_{ij} I_j\geq 0.
	\end{align*}
	Then it follows from Theorem 3.1 of Sun et al. \cite{sun2011effect} that all the $4n$ state variables remain non-negative for all time. Hence, $\mathbb{R}_{+}^{4n}$ is a positively invariant set for the system \eqref{EQ:eqn 2.1}.
\end{proof}

\begin{proposition}
    The solutions $(S_i(t), E_i(t), I_i(t), R_i(t))$ of the model \eqref{EQ:eqn 2.1} are ultimately and uniformly bounded in $\mathbb{R}_{+}^{4n}$.
\end{proposition}
\begin{proof}
    Let $N_i(t) = S_i(t)+E_i(t)+I_i(t)+R_i(t)$ be the entire human population at time $t$ in patch $i$ and in all patches, $N=\sum_{i=1}^{n} N_{i}$ be the entire population. Also we define $\Pi=\sum_{i=1}^{n} \Pi_{i}$ and $\mu=min_{1\leq i \leq n} \mu_{i}$.\\
    Now adding $4n$ equations we have 
    \begin{align*}
	&\frac{dN_i}{dt}=\Pi_i-\mu_i N_i-\delta_i I_i - \sum_{X=S,E,I,R} \Big(\sum_{j=1}^{n} m_{ji} X_i-\sum_{j=1}^{n} m_{ij} X_j\Big) \\
	\end{align*}
	Now, if we add $i$ from $1$ to $n$, we have,
	\begin{align*}
	&\frac{dN}{dt}=\sum_{i=1}^{n}\Big[\Pi_i-\mu_i N_i-\delta_i I_i - \sum_{X=S,E,I,R} \Big(\sum_{j=1}^{n} m_{ji} X_i-\sum_{j=1}^{n} m_{ij} X_j\Big)\Big]
	\end{align*}
	In the preceding equation, the double summation is vanished and since $I_i \leq N_i$ , it follows that
	\begin{align*}
	    \frac{dN}{dt}&\leq \sum_{i=1}^{n}\Pi_i-\sum_{i=1}^{n}\mu_i N_i\\
	    &\leq \sum_{i=1}^{n}\Pi_i-\sum_{i=1}^{n}min_{1\leq i \leq n} \mu_{i} N_i
	\end{align*}
	Hence by standard comparison theorem,
	\begin{align*}
	    N(t) \leq \Big\lbrace \frac{\sum_{i=1}^{n}\Pi_i}{\sum_{i=1}^{n}min_{1\leq i \leq n} \mu_{i}} \Big\rbrace \Longrightarrow N(t)\leq \frac{\Pi}{\mu}
 	\end{align*}
 	Therefore the solutions of the model \eqref{EQ:eqn 2.1} are ultimately and uniformly bounded in $\mathbb{R}_{+}^{4n}$.
\end{proof}

\begin{corollary}
The region $\Omega=\lbrace (S_i, E_i, I_i, R_i)\in \mathbb{R}_{+}^{4n}|N\leq\frac{\Pi}{\mu}\rbrace$ is invariant and attracting for system \eqref{EQ:eqn 2.1}.
\end{corollary}

\subsection{Stability of disease-free equilibrium and basic reproduction number}
Let $M=[m_{ij}]$ be the travel rate matrix for each compartments, is assumed to be 
irreducible. The matrix $M$ represents the human movement among patches.

Define $n \times n$ matrices
\begin{align*}
    \psi^{X}= \text{diag}(\mu_i^{X} + \sum_{j=1}^{n} m_{ji})-M
\end{align*}
where, $\mu_i^{S}=\mu_i$, $\mu_i^{E}=\gamma_i+\mu_i$,  $\mu_i^{I}=\sigma_i+\delta_i+\mu_i$ and  $\mu_i^{R}=\mu_i$.

For example,
\begin{align}\label{EQ:eqn 3.1}
 \psi^R =\begin{bmatrix}
    \mu_1 + \sum_{j=1}^{n} m_{j1} & -m_{12} & \cdots & - m_{1n}\\
    -m_{21} & \mu_2 + \sum_{j=1}^{n} m_{j2} &\cdots& -m_{2n}\\
    \vdots  & \vdots  & \ddots & \vdots\\
    - m_{n1}  &  - m_{n2} & \cdots &  \mu_n + \sum_{j=1}^{n} m_{jn}\\
    \end{bmatrix}
\end{align}
The matrices above are non-singular M-matrices since all off-diagonal entries are nonpositive (i.e., of the Z-sign pattern) and the sum of the entries in each column is positive and $\psi^{{X}}\geq 0$ \cite{berman1994nonnegative}, where $X=S,E,I,R$.
The phase at which no disease exists in the population is referred to as disease free equilibrium (DFE). For this we set $E_i=0, I_i=0$ in the equation \eqref{EQ:eqn 2.1}. Then the system of equation \eqref{EQ:eqn 2.1} reduces to:
\begin{eqnarray}\label{EQ:eqn 3.11}
\displaystyle{\frac{dS_i}{dt}} &=& \Pi_i - \mu_i S_i - \sum_{j=1}^{n} m_{ji} S_i + \sum_{j=1}^{n} m_{ij}  S_j, \\
\displaystyle{\frac{dR_i}{dt}} &=& - \mu_i  R_i - \sum_{j=1}^{n} m_{ji} R_i + \sum_{j=1}^{n} m_{ij} R_j , \nonumber
\end{eqnarray}
We take the following equations to calculate the DFE:
\begin{eqnarray}\label{EQ:eqn 3.12}
0 &=& \Pi_i - \mu_i S_i - \sum_{j=1}^{n} m_{ji} S_i + \sum_{j=1}^{n} m_{ij}  S_j,\\
0 &=& - \mu_i  R_i - \sum_{j=1}^{n} m_{ji} R_i + \sum_{j=1}^{n} m_{ij} R_j , \nonumber
\end{eqnarray}
This can be expressed as a matrix form:
\begin{align}\label{EQ:eqn 3.13}
    \psi^S S&=\Pi_i\\
    \psi^R R& =0\nonumber
\end{align}
where, $\psi^S=\text{diag}(\mu_i + \sum_{j=1}^{n} m_{ji})-M$, $\psi^R=\text{diag}(\mu_i + \sum_{j=1}^{n} m_{ji})-M$, $M$ is an irreducible movement matrix and $S=(S_1, S_2,......, S_n)^T$, $R=(R_1, R_2,......, R_n)^T$, $\Pi=(\Pi_1, \Pi_2,......, \Pi_n)^T$.

We can see that the matrices $\psi^S$ and $\psi^R$ have non-positive off-diagonal components and the total sum of each column's entries is positive. As a result, both $\psi^S$ and $\psi^R$ are non-singular M-matrices and have positive inverses \cite{berman1994nonnegative}. As a consequence, the second equation of \eqref{EQ:eqn 3.12} has trivial solution, while the first equation of \eqref{EQ:eqn 3.12} has a unique positive solution $S^0=(S_1^0, S_2^0,......, S_n^0)=(\psi^S)^{-1}\Pi$.

As a consequence, we get the following result.
\begin{theorem}
There exits a DFE $P^0=(S_1^0, 0, 0, 0, S_2^0, 0, 0, 0, ...., S_n^0, 0, 0, 0)\in \mathbb{R}_{+}^{4n}$ for the system \eqref{EQ:eqn 2.1} which is unique.
\end{theorem}

\subsubsection{Basic Reproduction Number}
The system \eqref{EQ:eqn 2.1} has a unique disease-free equilibrium (DFE) and is given by
\begin{align*}
    P^0=(S_1^0, 0, 0, 0, S_2^0, 0, 0, 0, ...., S_n^0, 0, 0, 0) \in \mathbb{R}_{+}^{4n}.
\end{align*}

Following \cite{senapati2019cholera}, the matrix $(F)$ of  new infection and the matrix $(V)$ of transition terms are given below:
\begin{align*}
F=\begin{bmatrix}
\begin{array}{cc}
0 & F_{11}\\
0 & 0\\
\end{array}
\end{bmatrix}
\text{ and }
V=\begin{bmatrix}
\begin{array}{cc}
V_{11} & 0\\
-V_{21} & V_{22}\\
\end{array}
\end{bmatrix}\\
\end{align*}
where,
\begin{align*}
 F_{11}&=\begin{bmatrix}
\beta_1 & m_{12} \alpha_2 & m_{13} \alpha_3 & \cdots & m_{1n}\alpha_n \\
m_{21} \alpha_1 & \beta_2 & m_{23} \alpha_3 & \cdots & m_{2n}\alpha_n \\
\vdots  & \vdots & \vdots & \ddots & \vdots  \\
m_{n1} \alpha_1 &  m_{n2} \alpha_2 &  m_{n3} \alpha_3 & \cdots & \beta_n \\
\end{bmatrix}
\\\\
    V_{11} &=\begin{bmatrix}
    (\gamma_1+\mu_1)+\sum_{j=1}^{n} m_{j1} & -(1-\xi_1)m_{12} & \cdots & -(1-\xi_1)m_{1n}\\
    -(1-\xi_2)m_{21} & (\gamma_2+\mu_2)+\sum_{j=1}^{n} m_{j2} &\cdots& -(1-\xi_2)m_{2n}\\
    \vdots  & \vdots  & \ddots & \vdots\\
    -(1-\xi_n)m_{n1}  & -(1-\xi_n)m_{n2} & \cdots & (\gamma_n+\mu_n)+\sum_{j=1}^{n} m_{jn}\\
    \end{bmatrix}\\\\
    V_{21}&=\begin{bmatrix}
    \gamma_1 & \xi_1 m_{12} & \cdots & \xi_1 m_{1n}\\
    \xi_2 m_{21} & \gamma_2 &\cdots& \xi_2 m_{2n}\\
    \vdots  & \vdots  & \ddots & \vdots\\
    \xi_n m_{n1}  &  \xi_n m_{n2} & \cdots &  \gamma_n\\
    \end{bmatrix}\\\\
    V_{22} &=\begin{bmatrix}
    (\sigma_1+\mu_1+\delta_1)+\sum_{j=1}^{n} m_{j1} & -(1-p_1)m_{12} & \cdots & -(1-p_1)m_{1n} \\
    -(1-p_2)m_{21} & (\sigma_2+\mu_2+\delta_2)+\sum_{j=1}^{n} m_{j2} & \cdots & -(1-p_2)m_{2n} \\
    \vdots  & \vdots  & \ddots & \vdots\\
    -(1-p_n)m_{n1} & -(1-p_n)m_{n2} & \cdots &  (\sigma_n+\mu_n+\delta_n)+\sum_{j=1}^{n} m_{jn}  \\
    \end{bmatrix}\\
\end{align*}
From the above we have seen that $F \geq 0$ and the Z-sign pattern is present in $V$. Since
\begin{align*}
V^{-1}=\begin{bmatrix}
\begin{array}{cc}
V_{11}^{-1} & 0\\
 V_{22}^{-1}V_{21} V_{11}^{-1}& V_{22}^{-1}\\
\end{array}
\end{bmatrix}\geq 0,\\
\end{align*}
V is a irreducible nonsingular M-matrix. Following \cite{van2002reproduction}, the basic reproduction number, denoted by $R_0$ is the spectral radius of the next generation matrix $FV^{-1}$, where,
\begin{align*}
    FV^{-1}=\begin{bmatrix}
\begin{array}{cc}
0 & F_{11}\\
0 & 0\\
\end{array}
\end{bmatrix}
\begin{bmatrix}
\begin{array}{cc}
V_{11}^{-1} & 0\\
 V_{22}^{-1}V_{21} V_{11}^{-1}& V_{22}^{-1}\\
\end{array}
\end{bmatrix}
\end{align*}
Therefore,
\begin{align}\label{basic reproduction number}
    R_0=\rho(FV^{-1})=\rho(F_{11}  V_{22}^{-1}V_{21} V_{11}^{-1})
\end{align}
where $\rho(X)$ represents the spectral radius of the matrix $X$.

To indicate that the expression for $R_0$ is dependent on the whole network, we refer to it as the domain basic reproduction number. According to Theorem 2 in \cite{van2002reproduction} that the DFE of the model \eqref{EQ:eqn 2.1} is locally asymptotically stable whenever $R_0<1$, while unstable if $R_0>1$.

In the specific case of $n=1$, the basic reproduction number has the explicit form 
\begin{align*}
    R_0=\frac{\beta \gamma}{(\gamma+\mu)(\sigma+\mu+\delta)}
\end{align*}
In the special case of no humans migrate between patches (i.e., $M=0$), the basic reproduction number $R_0$ defined in \ref{basic reproduction number} is represented by the maximum possible value of basic reproduction numbers $R_0$ in every patches. Therefore, $R_0=\max\limits_{i}\lbrace R_0^{(i)} \rbrace$\\
where
\begin{align}\label{Eq:3.13}
    R_0^{(i)}=\frac{\beta_i \gamma_i}{(\gamma_i+\mu_i)(\sigma_i+\mu_i+\delta_i)}        \text{ i=1,2,\dots,n}
\end{align}
\subsubsection{Bounds on $R_0$}
Let there is no human movements, i.e., $M=0$. Also let us assume that $\beta_i=\beta$, $\gamma_i=\gamma$, $\gamma_i+\mu_i=\gamma+\mu$ and $\sigma_i+\mu_i+\delta_i=\sigma+\mu+\delta$ for all $i$. Then,
\begin{align*}
    R_0=\frac{1}{\gamma+\mu}\rho(F_{11}  V_{22}^{-1}V_{21})
\end{align*}
From spectral properties of non-negative matrices, $R_0$ can be bounded above and below.
\begin{theorem}
Assume that $\beta_i=\beta$, $\gamma_i=\gamma$, $\gamma_i+\mu_i=\gamma+\mu$ and $\sigma_i+\mu_i+\delta_i=\sigma+\mu+\delta$ for all $i$ and $M^E=M^I=0$. Then 
\begin{align*}
    \min\limits_{i}\lbrace R_0^{(i)} \rbrace \leq R_0 \leq \max\limits_{i}\lbrace R_0^{(i)} \rbrace.
\end{align*}
\end{theorem}
\begin{proof}
The proof is motivated by similar outcomes in \cite{hsieh2007impact, salmani2006model} and the fact that $F_{11}  V_{22}^{-1}V_{21}$ is similar to $V_{21}F_{11}  V_{22}^{-1}$. Let $V_{22}^{-1}=Y=[y_{ij}]$, $\textbf{1}=(1,1,\dots,1)^T\in\mathbb{R}^n$ and $[1^TV_{21}F_{11}  V_{22}^{-1}]_i$ denote the sum of each entries in the $i^{th}$ column of $V_{21}F_{11}  V_{22}^{-1}$. Since the column sum of $V_{22}$ is $(\sigma+\mu+\delta)$, $\textbf{1}^T V_{22}=(\sigma+\mu+\delta)\textbf{1}^T$, giving $\textbf{1}^TY=\frac{1}{\sigma+\mu+\delta}\textbf{1}^T$. Therefore, the column sum of $Y$ is $\frac{1}{\sigma+\mu+\delta}$. Then
\begin{align*}
[1^TV_{21}F_{11}  V_{22}^{-1}]&=\beta_1 \gamma_1y_{1i}+\beta_2 \gamma_2y_{2i}+\dots+\beta_n \gamma_ny_{ni}\\
& \leq \max \limits_{i}\lbrace{\beta_i \gamma_i}\rbrace (y_{1i}+y_{2i}+\dots+y_{ni})\\
&=\max\limits_{i}\Big\lbrace \frac{\beta_i \gamma_i}{\sigma+\mu+\delta}\Big\rbrace.
\end{align*}
Similarly, $[1^TV_{21}F_{11}  V_{22}^{-1}]\geq \min\limits_{i}\Big\lbrace \frac{\beta_i \gamma_i}{\sigma+\mu+\delta}\Big\rbrace.$\\
The conclusion follows from the facts that $\rho(F_{11}  V_{22}^{-1}V_{21})$ lies between its minimum and maximum column sums and that $R_0=\rho(V_{22}^{-1}V_{21}F_{11})=\rho(F_{11}  V_{22}^{-1}V_{21})$.
\end{proof}
\subsubsection{Global stability of disease-free equilibrium}
\begin{theorem}
Assume that the movement matrix $M$ is irreducible. Then the following results exhibit for the model \eqref{EQ:eqn 2.1}\\
(1) If $R_0<1$, then the disease-free equilibrium (DFE) $P^0$  of \eqref{EQ:eqn 2.1} is globally asymptotically stable in the region $\Omega$.\\
(2) If $R_0>1$, then the disease-free equilibrium (DFE) $P^0$ of \eqref{EQ:eqn 2.1} is unstable and the model \eqref{EQ:eqn 2.1} is uniformly persistent.
\end{theorem}
\begin{proof}
To prove the above theorem, we follow the following procedure as in the proof of Theorem 5.1 in \cite{eisenberg2013cholera}.

Let us define $x = (E_1,E_2,....,E_n,I_1,I_2,.....,I_n)$. Following Exercise 1.2, in \cite{johnson1985matrix}, the matrices $FV^{-1}$ and $V^{-1}F$ have the same spectral radius. Thus $R_0 = \rho(V^{-1}F)=\rho(FV^{-1})$. Assume that $b$ is the left eigenvector of $V^{-1}F$, which corresponds to the eigenvalue $R_0$. This implies $b^TV^{-1}F = R_0 b^T$.

Let us consider the following Lyapunov function
\begin{align}\label{EQ:eqn 3.14}
    L=(b^TV^{-1}x)
\end{align}
Differentiating $L$ along \eqref{EQ:eqn 2.1} gives,

\begin{align}\label{EQ:eqn 3.15}
    L'&=b^TV^{-1}x'\nonumber\\
       &\leq b^T V^{-1}(F-V)x\nonumber\\
       &=b^T V^{-1}F x - b^T x\nonumber\\
       &=(R_0-1)b^T x \leq 0,  \text{ if } R_0<1.
\end{align}

Using the irreducibility of the matrix $M$, it can be checked that the singleton $\lbrace P^0 \rbrace$ is the unique invariant set where $L'=0$. Therefore by using LaSalle’s invariance principle \cite{la1976stability}, $P^0$ is globally asymptotically stable in $\Omega$.

Alternatively, if $R_0>1$ and $x>0$, it follows that $(R_0-1)b^Tx >0$. This combination of inequality and continuity imply that $L'> 0$ for a small neighborhood of $P^0$ in int($\Omega$). That is, for $R_0>1$, any solution which is sufficiently close to $P^0$ will step away from $P^0$. Based on the results of \cite{li1999global}, and using the irreducibility of the matrix $M$, the instability of $P^0$ indicates the uniform persistence of the model \eqref{EQ:eqn 2.1}.
\end{proof}

\section{Numerical simulations}\label{section4}
We consider the infection is originated in patch 1 i.e, few people are infected in patch 1 while other patches have no infection in the beginning. Initial conditions are $(S_1(0), E_1(0), I_1(0), R_1(0))=(100000, 100, 10, 0)$ and $(S_i(0), E_i(0), I_i(0), R_i(0))=(100000, 0, 0, 0)$ for i=2,3,4,5. Some fixed parameters are reported in Table \ref{Tab:table_parameters} and some other fixed parameters are $d_{12}=d_{21}=100$, $d_{13}=d_{31}=110$, $d_{14}=d_{41}=120$, $d_{15}=d_{51}=130$, $d_{23}=d_{32}=140$, $d_{24}=d_{42}=150$, $d_{25}=d_{52}=160$, $d_{34}=d_{43}=170$, $d_{35}=d_{53}=180$, $d_{45}=d_{54}=190$ and $(\nu_1, \nu_2, \nu_3, \nu_4, \nu_5)=(7, 6, 5, 4, 3)$. These parameters are kept constant throughout the numerical simulation section.

\begin{table*}
	\caption{Parameter values and their units used in the numerical simulation}
	\label{Tab:table_parameters}
	\centering
	\begin{tabular}{|c|c|c|}
		\hline
		Parameter & value/Range & unit \\
		\hline
		$\mu_i$ & 0.00032 & week$^{-1}$ \\
		$\Pi_i$ & $N_i \times \mu_i$ & Person week$^{-1}$ \\
		$\beta_i$ & (0,1) & week$^{-1}$ \\
		$\alpha_i$ & (0,1) & week$^{-1}$ \\
		$\gamma_i$ & 0.15 & week$^{-1}$\\
		$\xi_i$ & 0.3 & week$^{-1}$ \\
		$\sigma_i$ & 0.09 & week$^{-1}$ \\
		$\delta_i$ & 0.05 & week$^{-1}$ \\
		$p_i$ & 0.001 & week$^{-1}$ \\
		$\epsilon$ & (0,1) & Kilometers Person$^{-1}$ week$^{-1}$ \\
		$\eta$ & 0.25 & Unitless \\
		$\nu$ & $[1,\infty)$ & Person$^{-1}$ \\
		$d_{ij}$ & (0,5000) & Kilometers \\
		\hline
	\end{tabular}
\end{table*}

\subsection{Effect of network topology}
In this section, we simulate the model \eqref{EQ:eqn 2.1} to analyze different mobility patterns. Three different structures of the underlying network is considered, namely, fully connected network, ring of patches and a star like structure. Without loss of generality, we choose the number of patches to be five (see Fig. \ref{fig:patch_structures}). The fixed parameters are taken from Table \ref{Tab:table_parameters} and other parameters are 
$(\beta_1, \beta_2, \beta_3, \beta_4, \beta_5)=(0.3, 0.24, 0.18, 0.12, 0.06)$, $(\alpha_1, \alpha_2, \alpha_3, \alpha_4, \alpha_5)=(0.4, 0.35, 0.3, 0.25, 0.2)$ and $\epsilon = 0.005$.

\begin{figure}
    \centering
    \includegraphics[width=1.0\textwidth]{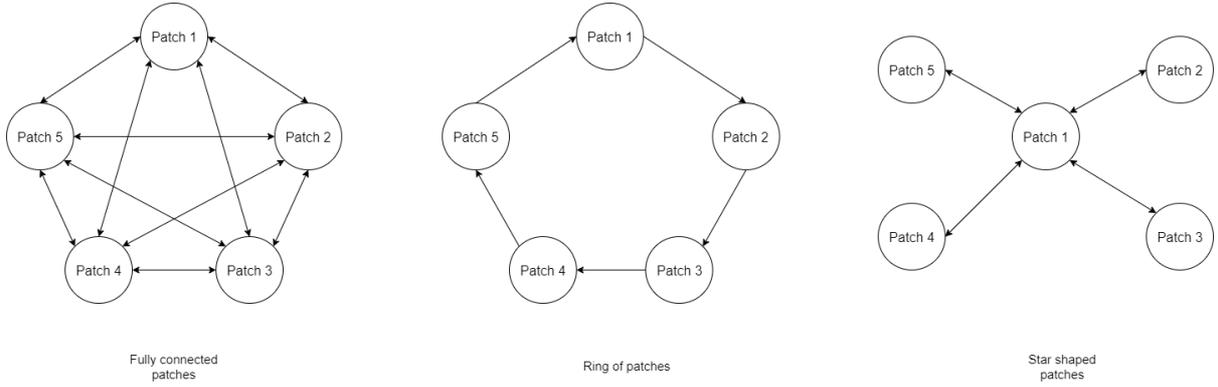}
    \caption{Different patch structures examined in numerical simulation.}
    \label{fig:patch_structures}
\end{figure}

\subsubsection{Scenario I: Fully connected patches}
The dynamics of infected population in fully connected metapopulation structure is depicted in Fig. \ref{fig:patch_structures_ts}(A). The proposed model \eqref{EQ:eqn 2.1} is simulated with the fixed parameters. It can be observed that infected populations in patch 1 and patch 2 show similar trend of outbreak. However, patch 2 experiences a delayed outbreak as that of patch 1. As the epidemic is originated in patch 1, the subsequent patches have delayed outbreaks. It is also noted that the peak of the outbreaks are decreasing as the transmission rates vary in the same order.

\subsubsection{Scenario II: Ring of patches}
In this scenario, the underlying network is assumed to be a ring where migration occur in a cyclic pattern. We consider a confined case of system \eqref{EQ:eqn 2.1} by assuming that the patches are organized in a ring. Individuals in a given patch $i$ will only move to patch $i+1$ and from patch $n$ to patch $1$, which is referred to as one-way migration. In this case all movement rates are zero except
\begin{align*}
    m_{1n}>0, m_{21}>0, m_{32}>0,......., m_{nn-1}>0.
\end{align*}
Then the system \eqref{EQ:eqn 2.1} reduces to 
\begin{eqnarray}\label{EQ:eqn 4.1}
\displaystyle{\frac{dS_i}{dt}} &=& \Pi_i- \beta_i \frac{I_i}{N_i} S_i - \mu_i S_i - m_{i+1 i} S_i + m_{ii-1} \Big(1-\frac{\alpha_{i-1} I_{i-1}}{N_{i-1}}\Big) S_{i-1},\nonumber \\
\displaystyle{\frac{dE_i}{dt}} &=& \beta_i \frac{I_i}{N_i} S_i -(\gamma_i+\mu_i)E_i - m_{i+1 i} E_i +  m_{ii-1} \frac{\alpha_{i-1} I_{i-1}}{N_{i-1}} S_{i-1} + (1-\xi_i)m_{ii-1} E_{i-1}, \\
\displaystyle{\frac{dI_i}{dt}} &=& \gamma_i E_i - (\sigma_i+ \mu_i +\delta_i)I_i - m_{i+1i} I_i + \xi_i m_{ii-1} E_{i-1} + (1-p_i) m_{ii-1} I_{i-1}, \nonumber \\
\displaystyle{\frac{dR_i}{dt}} &=& \sigma_i I_i - \mu_i  R_i - m_{i+1i} R_i + m_{ii-1} R_{i-1} + p_i m_{ii-1} I_{i-1}, \nonumber
\end{eqnarray}
where 
\begin{align*}
    & \text{if } i=1, \text{then } i-1=n\\
    & \text{if } i=n, \text{then } i+1=1
\end{align*}

The time evolution of infected populations in this case are portrayed in Fig. \ref{fig:patch_structures_ts}(B). We observe that patch 1 undergoes recurrent epidemic i.e, multiple peaks are observed. Additionally, the magnitude of the peaks are higher than that of the fully connected case. The peaks for infected compartments of different patches are in decreasing order of magnitude as the transmission rates vary in the same order.

\begin{figure}
    \centering
    \includegraphics[width=0.32\textwidth]{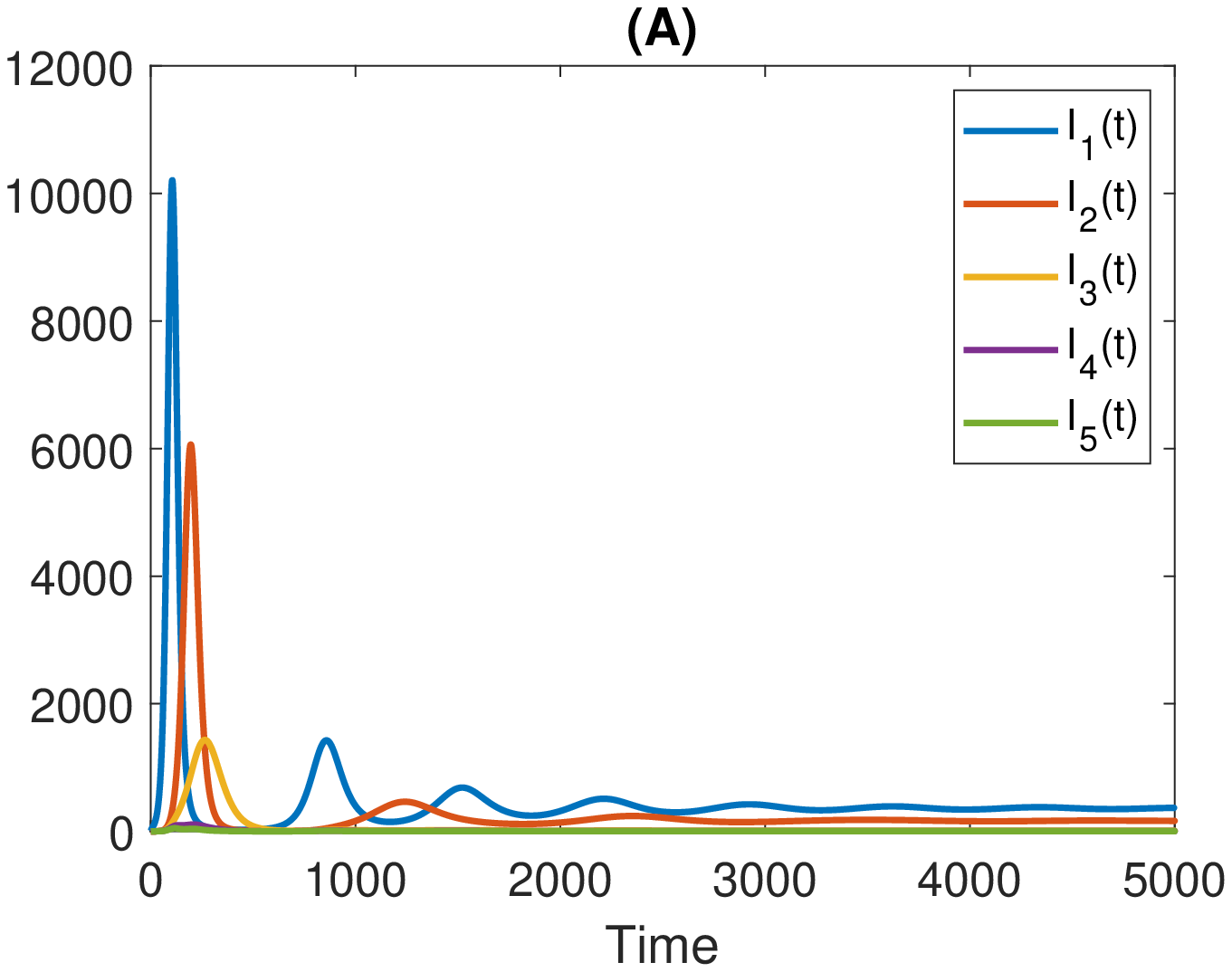}
    \includegraphics[width=0.32\textwidth]{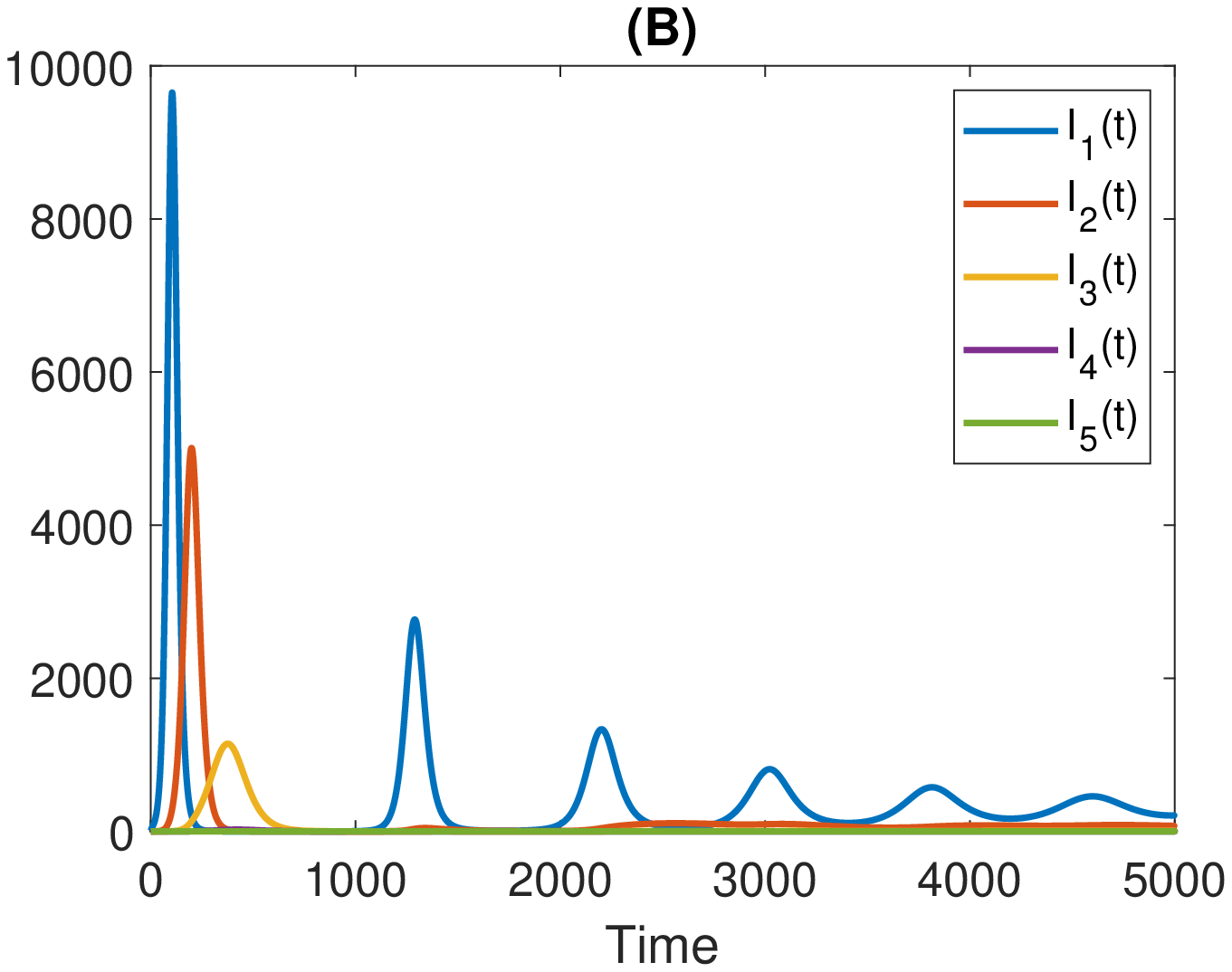}
    \includegraphics[width=0.32\textwidth]{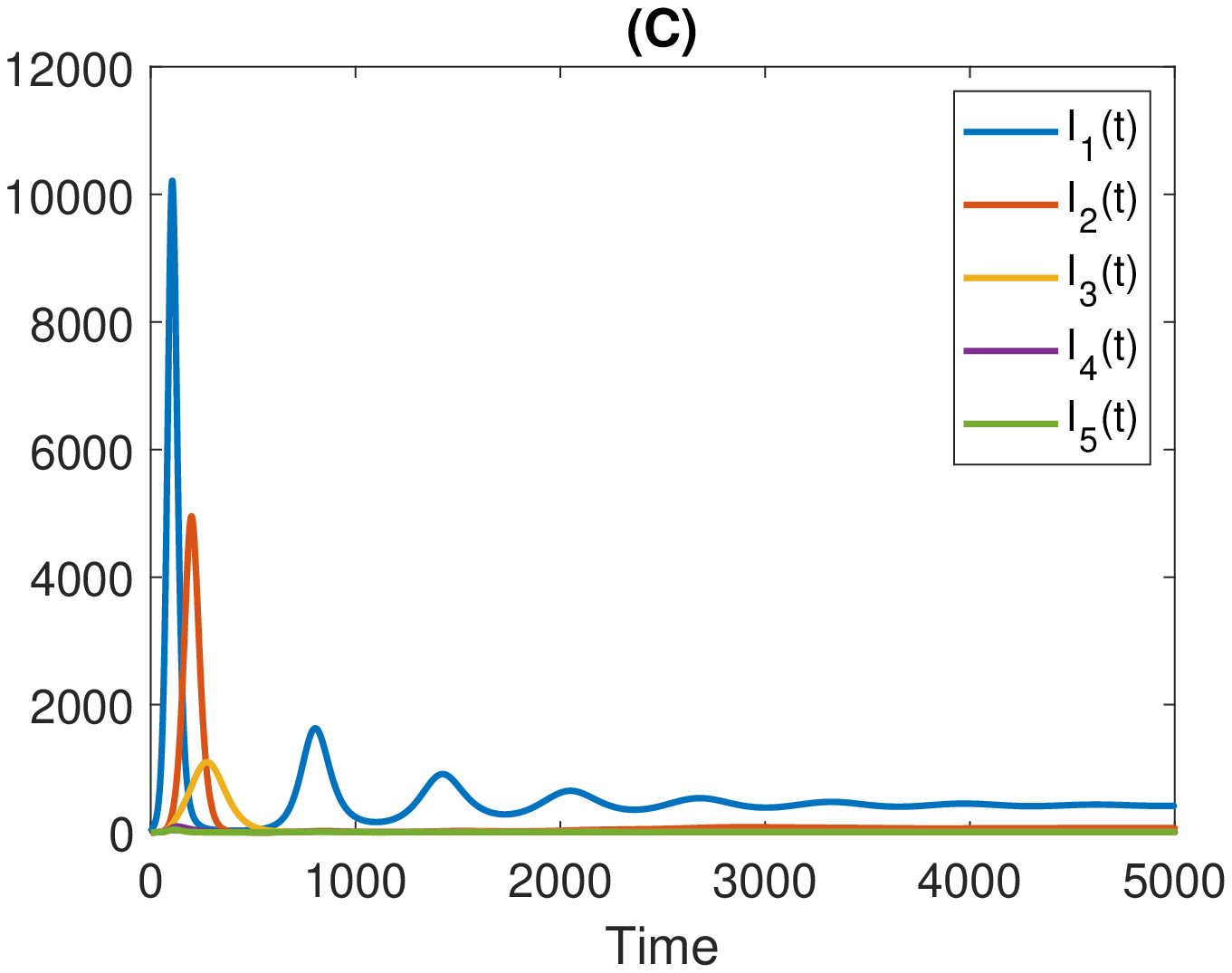}
    \caption{Time evolution of infected compartments for different patch structures (A) fully connected, (B) Ring of patches and (C) star network respectively.}
    \label{fig:patch_structures_ts}
\end{figure}

\subsubsection{Scenario III: Star shaped patches}
A star graph is a graph with one hub node in the center and remaining $ n - 1 $ nodes connected to the center hub node. The number of connections or edges in a star graph of $n$ vertices is $n-1$, where any node except the hub node has only one connection with the remaining nodes and the hub node has $n-1$ connections. Without loss of generality, we choose node $1$ as a hub node. In this case all movement rates are zero except
\begin{align*}
    & m_{1n}>0, m_{12}>0, m_{13}>0,......., m_{1 n-1}>0\\
    & m_{n1}>0, m_{21}>0, m_{31}>0,......., m_{n-1 1}>0
\end{align*}
Then the system \eqref{EQ:eqn 2.1} reduces to 
\begin{eqnarray}\label{EQ:eqn 5.1}
 \displaystyle{\frac{dS_1}{dt}}  &=& \Pi_1- \beta_1 \frac{I_1}{N_1} S_1 - \mu_1 S_1 - \sum_{j=1}^{n} m_{j1} S_1 + \sum_{j=1}^{n} m_{1j} (1-\frac{\alpha_j I_j}{N_j}) S_j,\nonumber \\
\displaystyle{\frac{dE_1}{dt}} &=& \beta_1 \frac{I_1}{N_1} S_1 -(\gamma_1+\mu_1)E_1 - \sum_{j=1}^{n} m_{j1} E_1 + \sum_{j=1}^{n} m_{1j} \frac{\alpha_j I_j}{N_j} S_j + \sum_{j=1}^{n} (1-\xi_1)m_{1j} E_j, \nonumber\\
\displaystyle{\frac{dI_1}{dt}} &=& \gamma_1 E_1 - (\sigma_1+ \mu_1 +\delta_1)I_1 - \sum_{j=1}^{n} m_{j1} I_1 + \sum_{j=1}^{n} \xi_1 m_{1j} E_j + \sum_{j=1}^{n} (1-p_1) m_{1j} I_j, \nonumber \\
\displaystyle{\frac{dR_1}{dt}} &=& \sigma_1 I_1 - \mu_1  R_1 - \sum_{j=1}^{n} m_{j1} R_1 + \sum_{j=1}^{n} m_{1j} R_j + \sum_{j=1}^{n} p_1 m_{1j} I_j, \\
\displaystyle{\frac{dS_i}{dt}} &=& \Pi_i- \beta_i \frac{I_i}{N_i} S_i - \mu_i S_i - m_{1 i} S_i + m_{i1} \Big(1-\frac{\alpha_{1} I_{1}}{N_{1}}\Big) S_{1},\nonumber \\
\displaystyle{\frac{dE_i}{dt}} &=& \beta_i \frac{I_i}{N_i} S_i -(\gamma_i+\mu_i)E_i - m_{1 i} E_i +  m_{i1} \frac{\alpha_{1} I_{1}}{N_{1}} S_{1} + (1-\xi_i)m_{i1} E_{1}, \nonumber\\
\displaystyle{\frac{dI_i}{dt}} &=& \gamma_i E_i - (\sigma_i+ \mu_i +\delta_i)I_i - m_{1i} I_i + \xi_i m_{i1} E_{1} + (1-p_i) m_{i1} I_{1}, \nonumber \\
\displaystyle{\frac{dR_i}{dt}} &=& \sigma_i I_i - \mu_i  R_i - m_{1i} R_i + m_{i1} R_{1} + p_i m_{i1} I_{1},
\nonumber
\end{eqnarray}
where the first $4$ equation for $i=1$ and last $(4n-4)$ equations for $2 \leq i \leq n$.

The dynamics of infected population in star shaped patches is depicted in Fig. \ref{fig:patch_structures_ts}(C). In this case, the time series behave similarly as that of the fully connected case. However, the prevalence of infection in the patches 1 through 4 are different from the fully connected case.

Based on these observations, we can conclude that the underlying network structure is an important factor in metapopulation modelling. Different network topology may induce significant changes in the behavior of the epidemic disease. MATLAB codes to generate these time series are provided in GitHub \footnote{\url{https://github.com/indrajitg-r/Metapop_IDT}}.

\subsection{Effect of infection during travel}
The proposed model \eqref{EQ:eqn 2.1} reduces to a model without infection during travel by putting $\alpha_i=0$, $\xi_i=0$ and $p_i=0$. Therefore, the system without infection during travel becomes

\begin{eqnarray}\label{EQ:eqn 6.1}
\displaystyle{\frac{dS_i}{dt}} &=& \Pi_i- \beta_i \frac{I_i}{N_i} S_i - \mu_i S_i - \sum_{j=1}^{n} m_{ji} S_i + \sum_{j=1}^{n} m_{ij} S_j,\nonumber \\
\displaystyle{\frac{dE_i}{dt}} &=& \beta_i \frac{I_i}{N_i} S_i -(\gamma_i+\mu_i)E_i - \sum_{j=1}^{n} m_{ji} E_i + \sum_{j=1}^{n}m_{ij} E_j, \\
\displaystyle{\frac{dI_i}{dt}} &=& \gamma_i E_i - (\sigma_i+ \mu_i +\delta_i)I_i - \sum_{j=1}^{n} m_{ji} I_i + \sum_{j=1}^{n} m_{ij} I_j, \nonumber \\
\displaystyle{\frac{dR_i}{dt}} &=& \sigma_i I_i - \mu_i  R_i - \sum_{j=1}^{n} m_{ji} R_i + \sum_{j=1}^{n} m_{ij} R_j, \nonumber
\end{eqnarray}

The parameter values are same as previous section except
$(\beta_1, \beta_2, \beta_3, \beta_4, \beta_5) =\\ 
(0.14, 0.12, 0.1, 0.08, 0.06)$, $(\alpha_1, \alpha_2, \alpha_3, \alpha_4, \alpha_5)=(0.9, 0.8, 0.7, 0.6, 0.5)$, and $\epsilon = 0.3$. Using this parameter set we simulate the models \eqref{EQ:eqn 2.1} to get the infected human population with infection during travel. On the other hand, \eqref{EQ:eqn 6.1} is simulated using this parameter set along with $\alpha_i=0$, $\xi_i=0$ and $p_i=0$ to generate time series without infection during travel. The five infected compartments corresponding to five patches are reported in Fig. \ref{fig:with_without_IDT}. From this figure, it can be observed that the endemic equilibrium is stable when there is infection during travel whereas the DFE is stable when there is no infection during travel. Therefore, it can be inferred that the infection during travel has the potential to alter stability of the system from disease-free to an endemic equilibrium. This reinforces that infection during travel may play a crucial role in disease transmission and persistence. 

\begin{figure}
    \centering
    \includegraphics[width=0.32\textwidth]{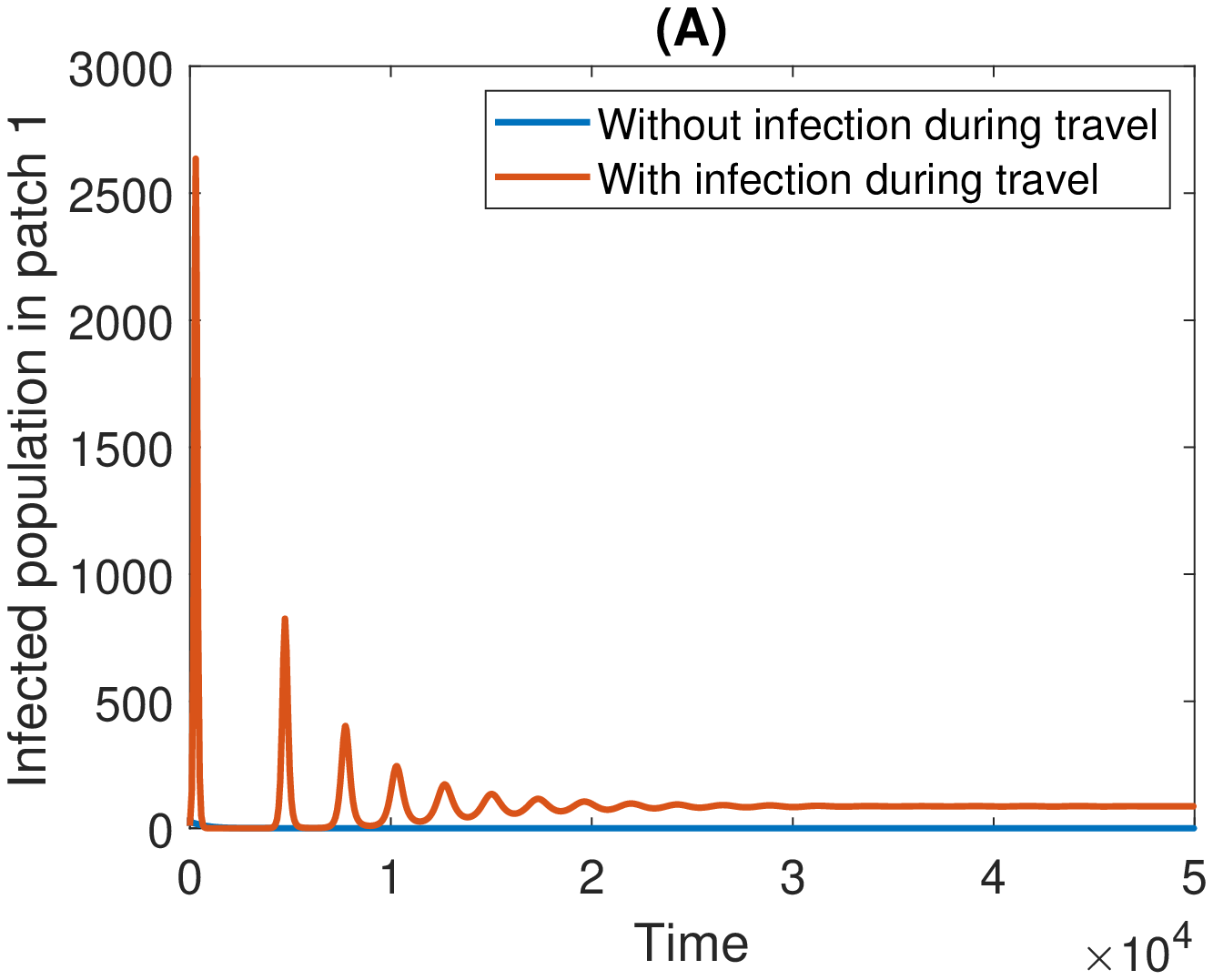}
    \includegraphics[width=0.32\textwidth]{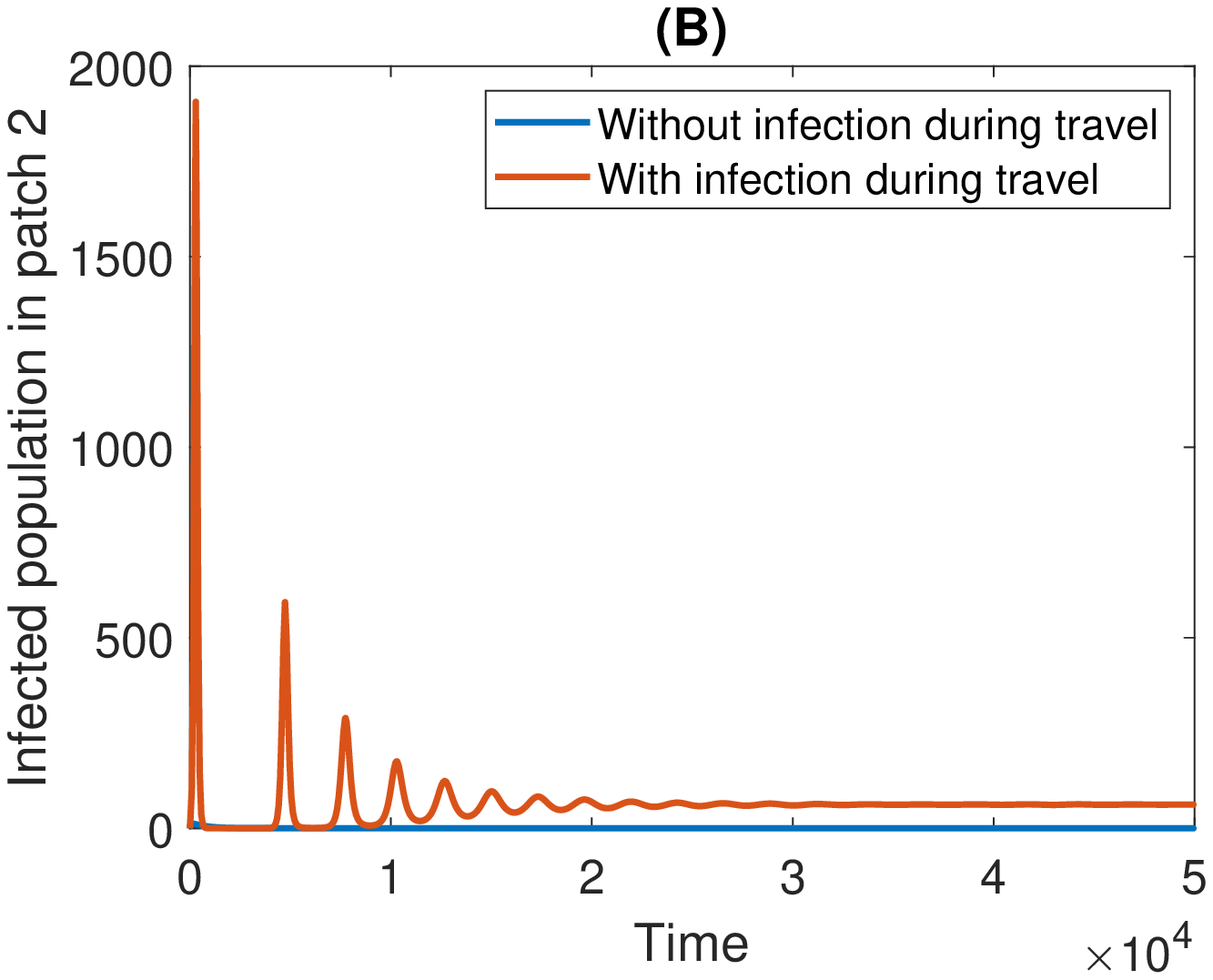}
    \includegraphics[width=0.32\textwidth]{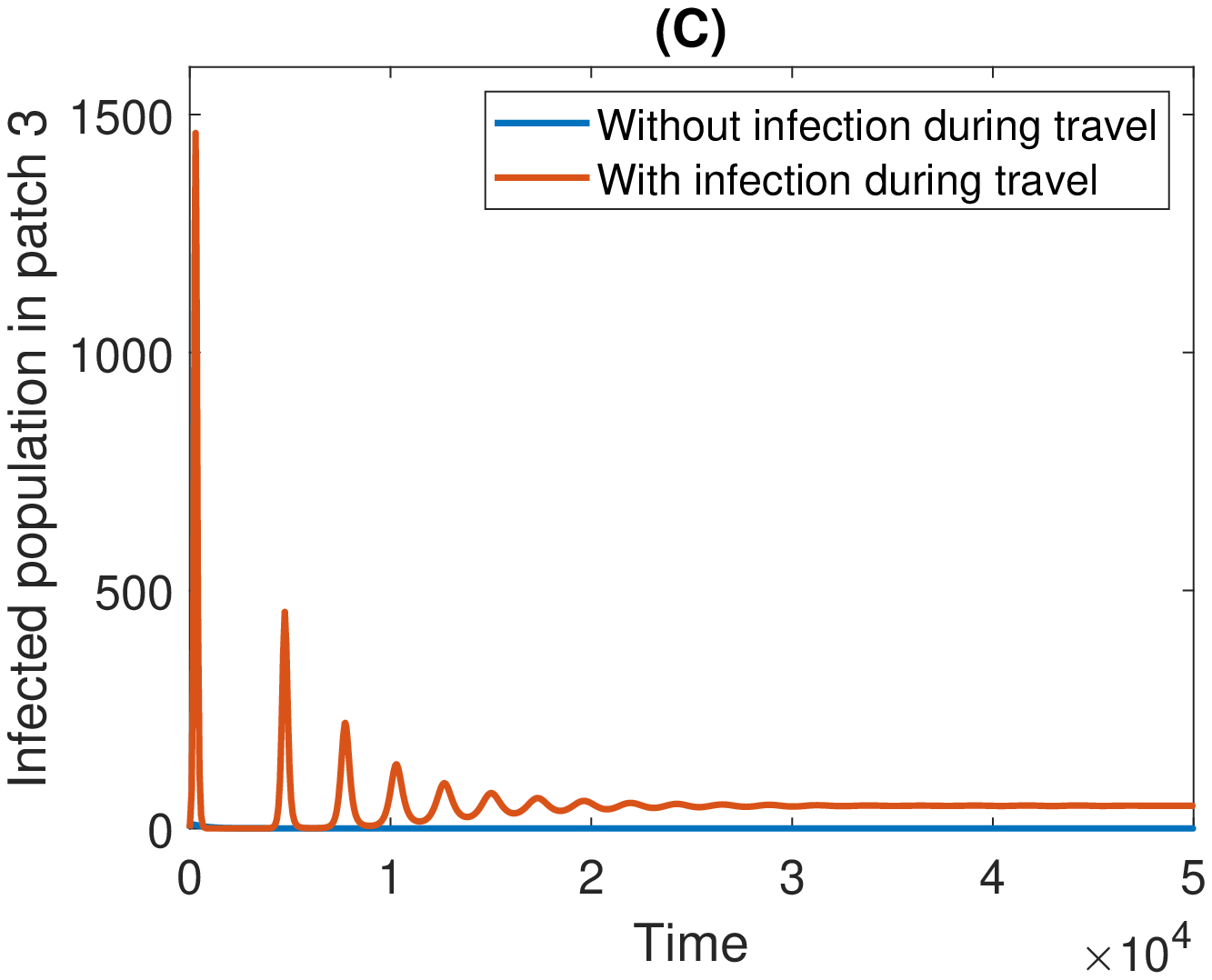}
    \includegraphics[width=0.32\textwidth]{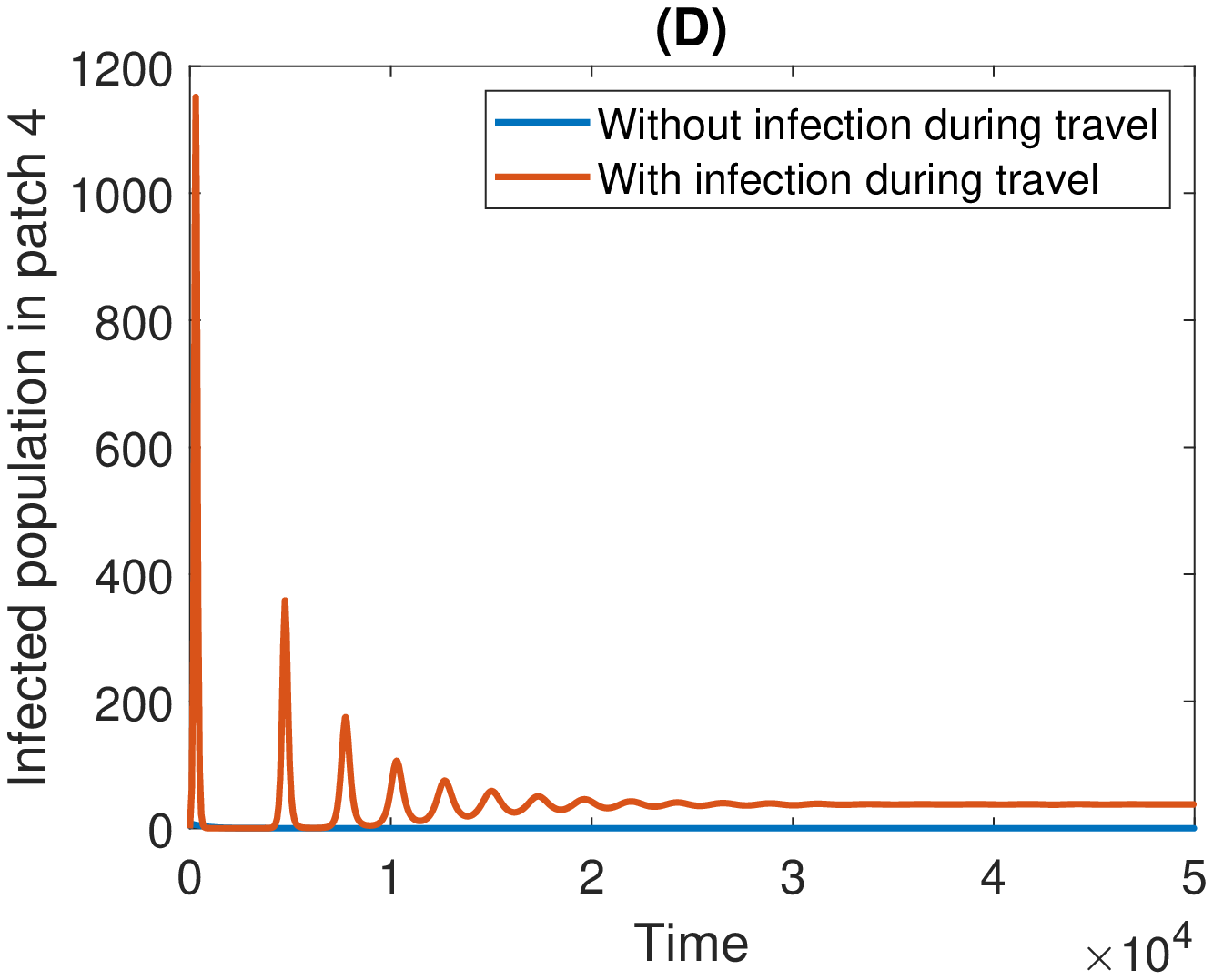}
    \includegraphics[width=0.32\textwidth]{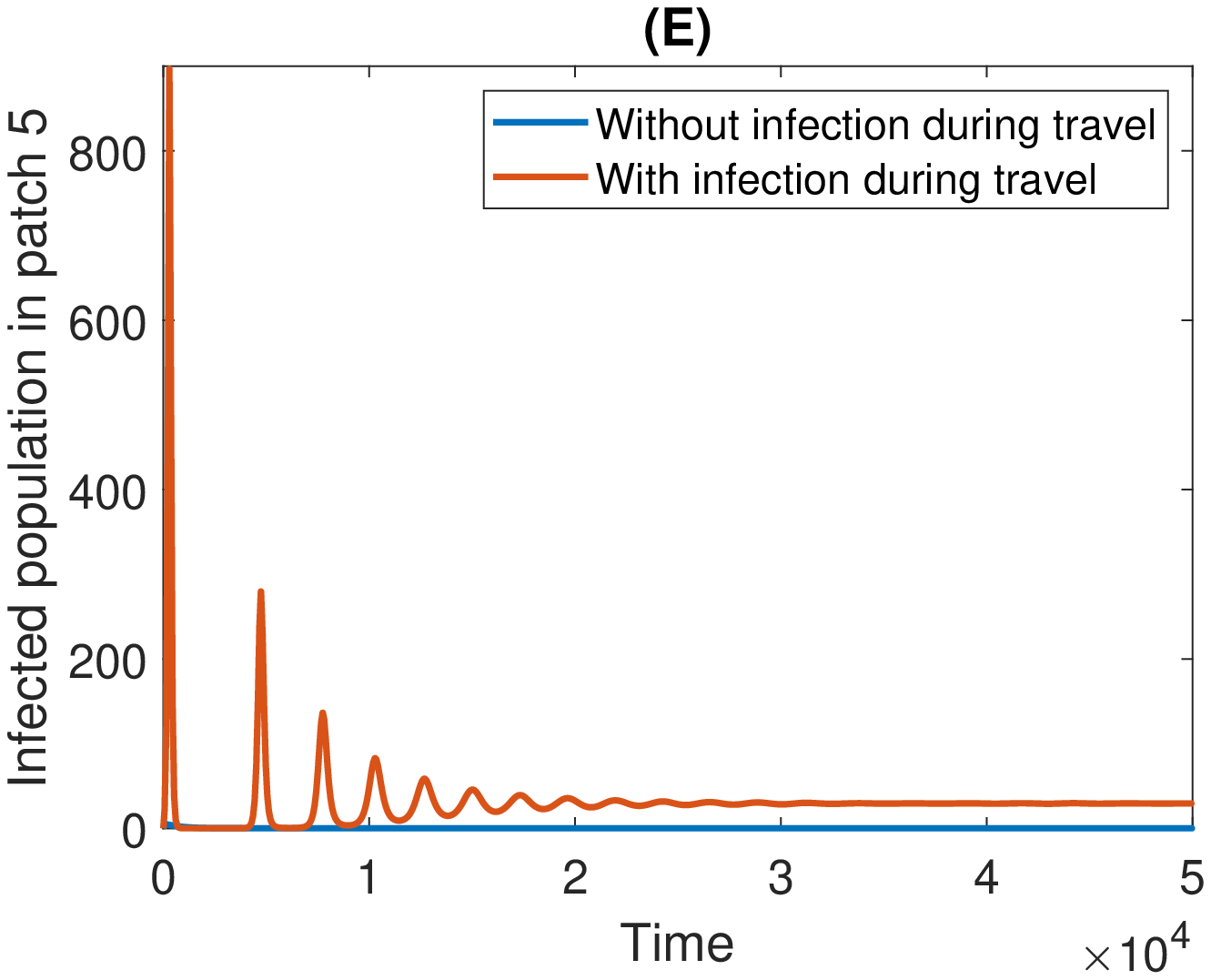}
    \caption{Behaviour of with and without infection during travel in model \eqref{EQ:eqn 2.1}. The basic reproduction number corresponding to infection during travel is $R_0 = 1.1053$ while for no infection during travel $R_0 = 0.7817$.}
    \label{fig:with_without_IDT}
\end{figure}

To investigate the effects of transmission coefficients during travel ($\alpha_i$), we simulate the system \eqref{EQ:eqn 2.1} with variable $\alpha_i$'s ($0.1 \leq \alpha_i \leq 0.9$). The total number of infected people after 100 weeks is calculated for each cases. While varying $\alpha_1$ we keep other $\alpha_i$'s to be zero. This is done to distinguish individual effects of $\alpha_i$'s. The boxplots of infected population for $\alpha_i$'s are depicted in Fig. \ref{fig:alpha_box_plots}. All the coefficients show similar trends of total infected cases. However, the scales of total infection are slightly different for different $\alpha_i$'s. This also indicate that the transmission during transport has a significant role in disease transmission.

\begin{figure}
    \centering
    \includegraphics[width=0.32\textwidth]{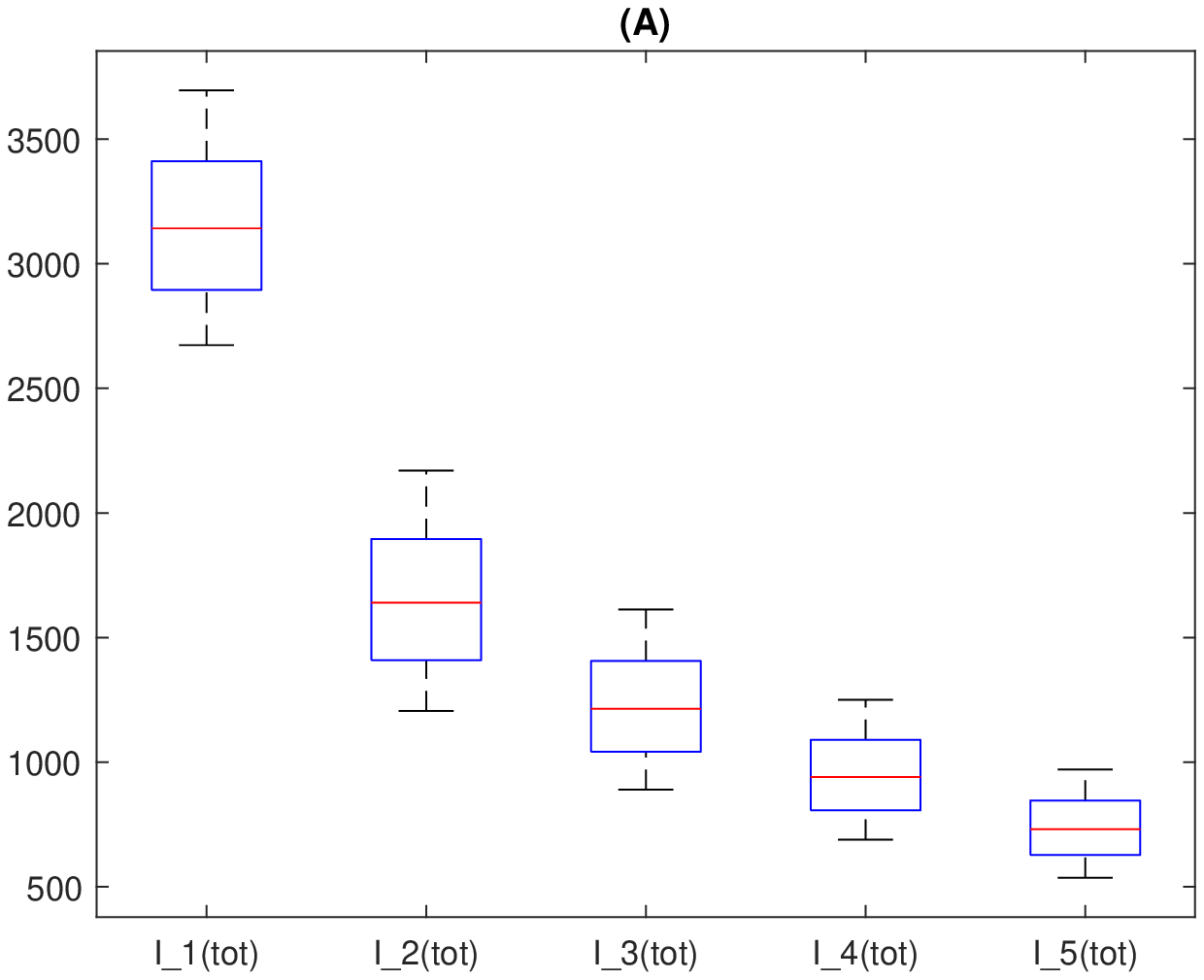}
    \includegraphics[width=0.32\textwidth]{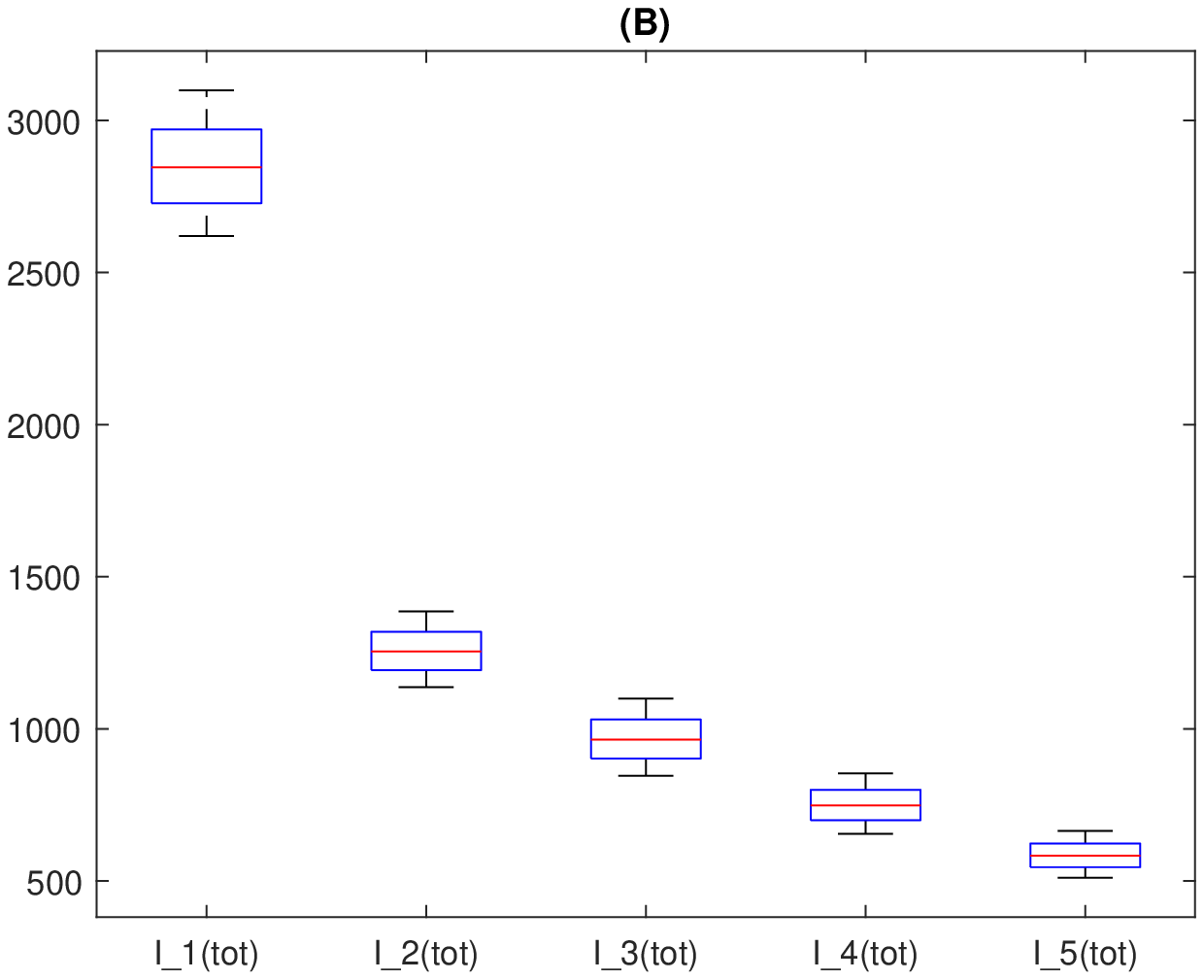}
    \includegraphics[width=0.32\textwidth]{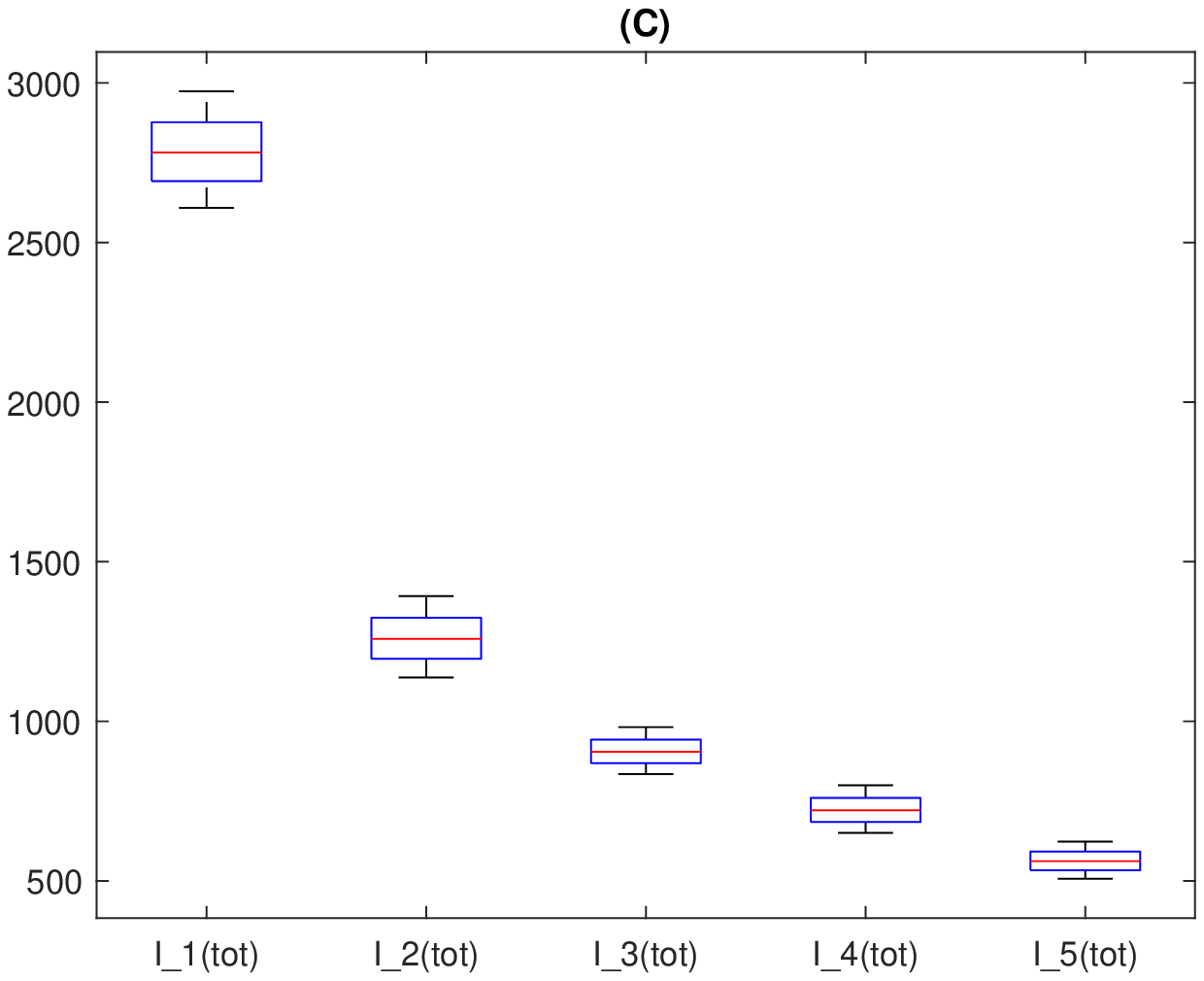}
    \includegraphics[width=0.32\textwidth]{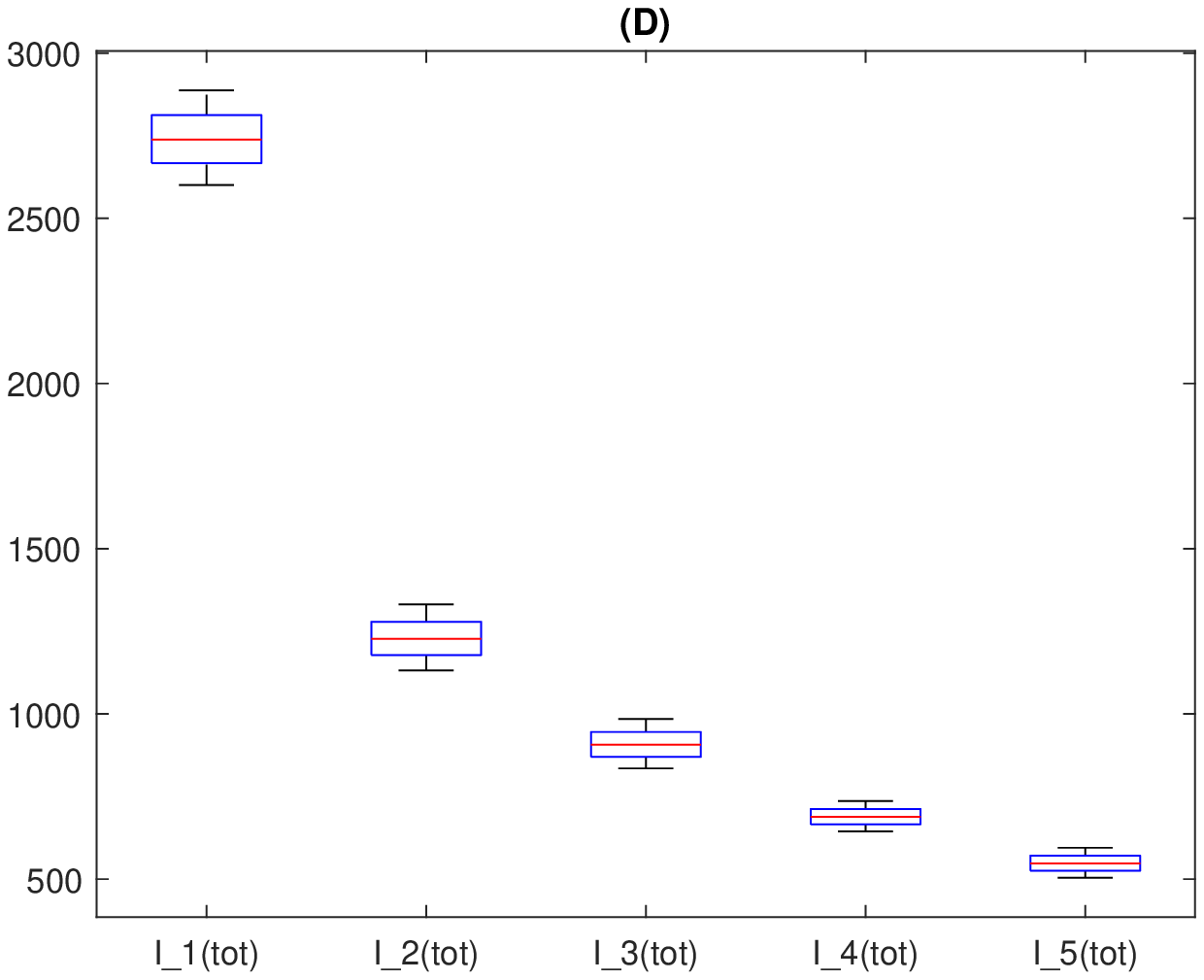}
    \includegraphics[width=0.32\textwidth]{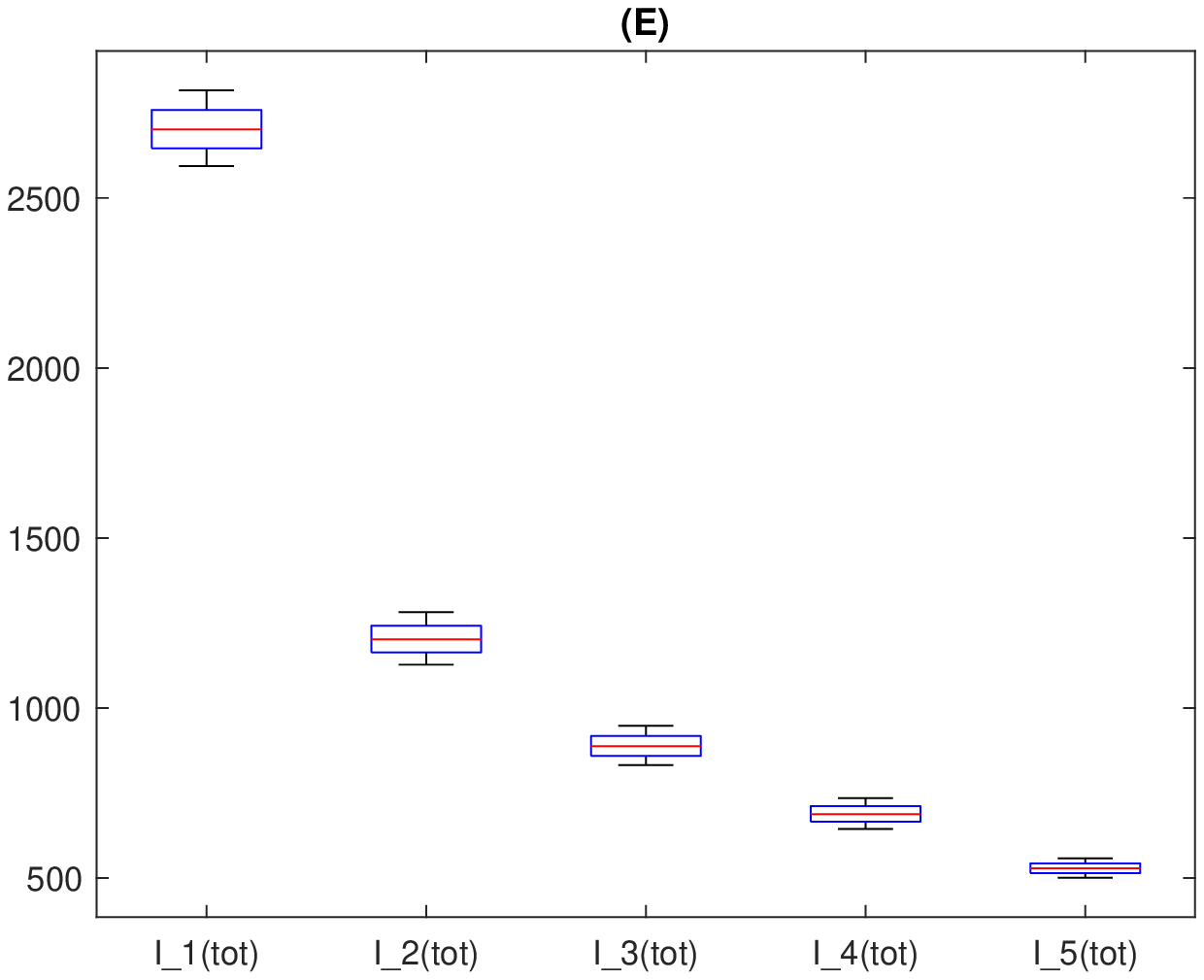}
    \caption{Box plots of total infectious people for different values of the transmission coefficient during travel.}
    \label{fig:alpha_box_plots}
\end{figure}

\subsection{Effect of coupling strength and transmission coefficients}
We draw contour plots of basic reproduction number ($R_0$) with respect to the coupling strength ($\epsilon$) and transmission coefficient during travel in patch 1 ($\alpha_1$) to investigate their effects. The contour plot is depicted in left panel of Fig. \ref{fig:countour_R0}. Nonlinear relations are observed for the variation in $\epsilon$, initially $R_0$ increases with increase in $\epsilon$ and after a certain value is crossed the relation gets reversed. From this figure, it is also observed that $\epsilon$ has a dominant effect on $R_0$ over $\alpha_1$. On the other hand, transmission coefficient in patch 1 ($\beta_1$) has a positive impact on $R_0$ (see Fig. \ref{fig:countour_R0}(B)). However, increasing $\epsilon$ shows a similar effect as that of Fig. \ref{fig:countour_R0}(A). 
 
\begin{figure}
    \centering
    \includegraphics[width=0.45\textwidth]{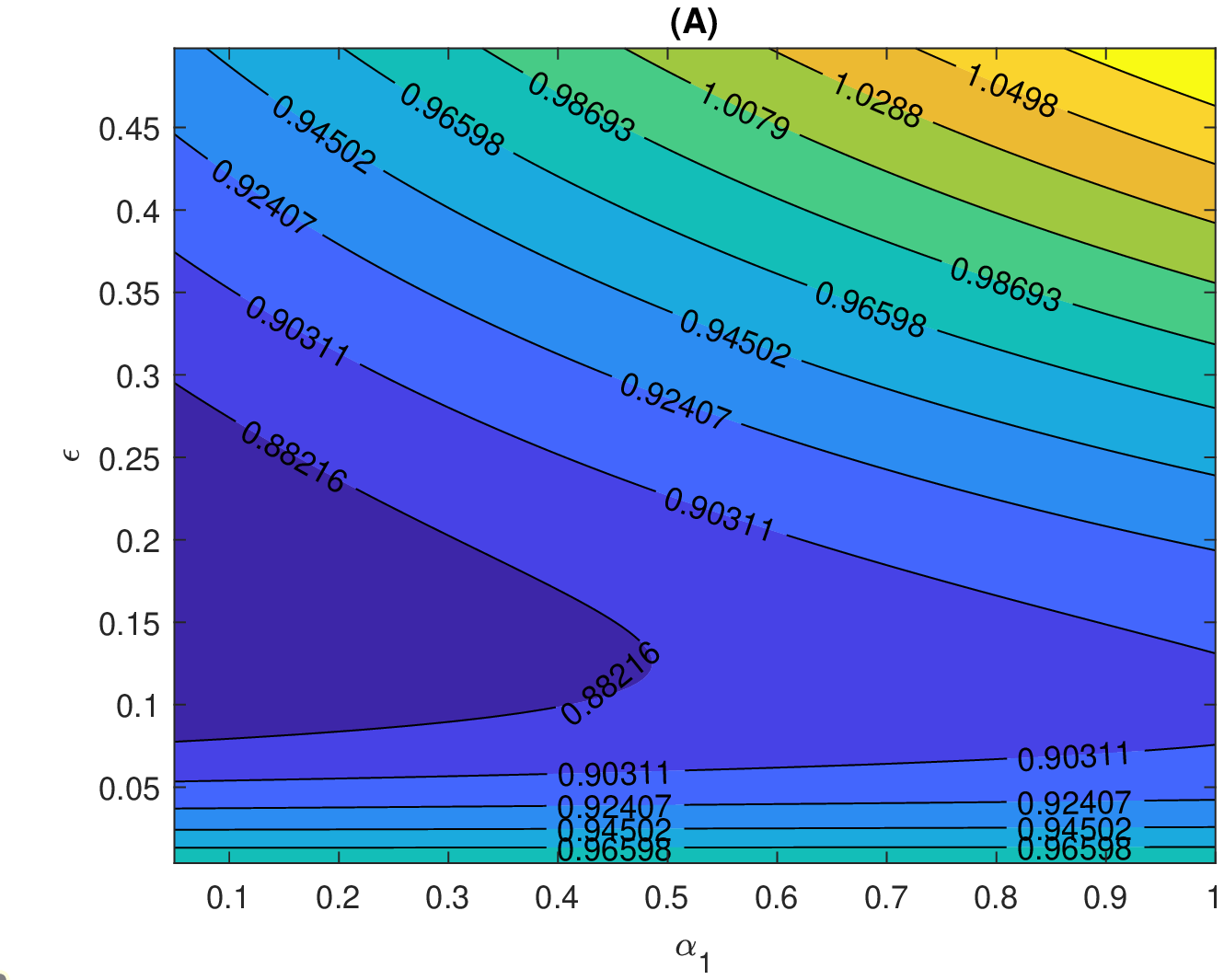}
    \includegraphics[width=0.45\textwidth]{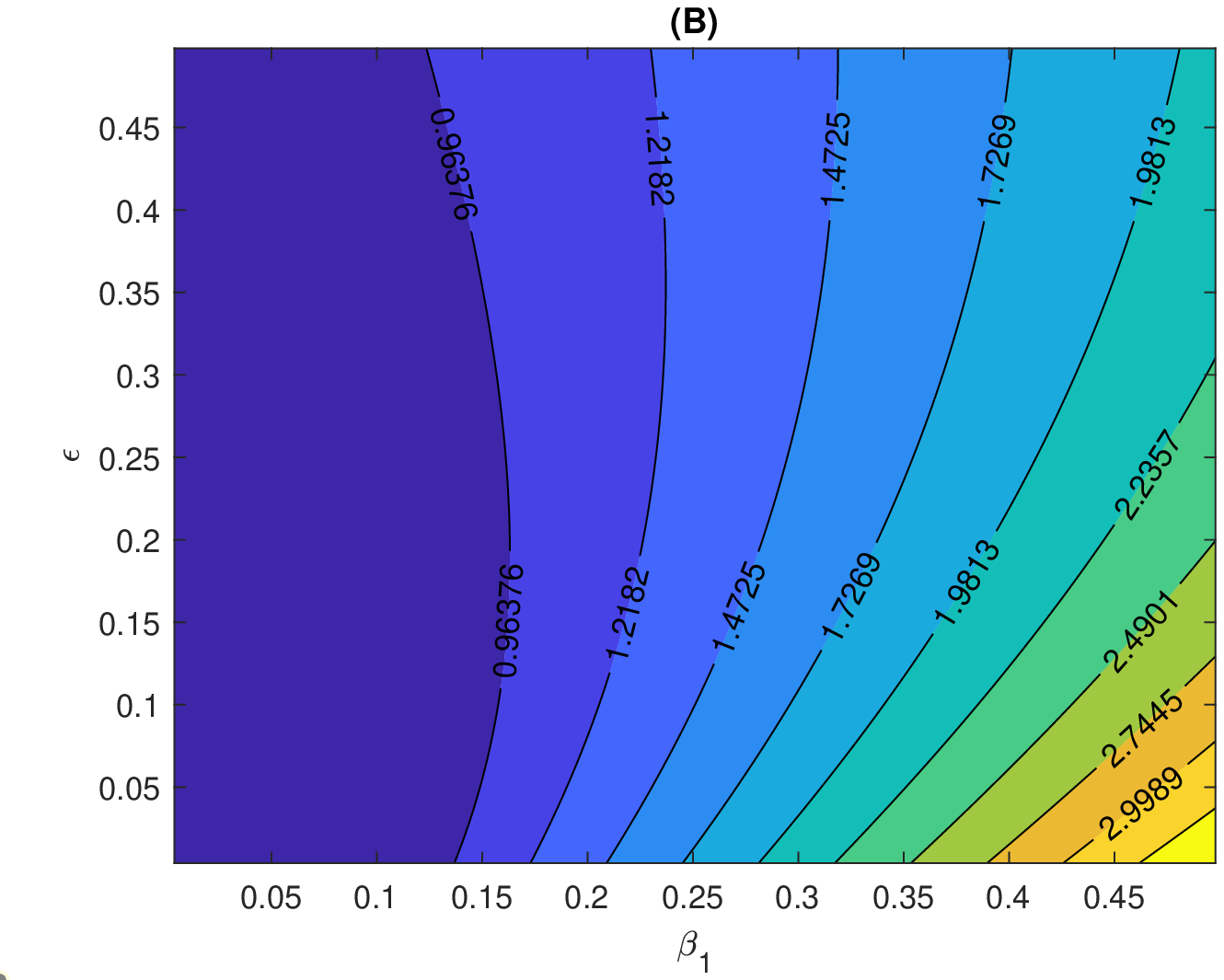}
    \caption{Contour plots of the basic reproduction number with respect to (A) $\epsilon$ and $\alpha_1$ and (B) $\epsilon$ and $\beta_1$. The fixed parameters are taken from Table \ref{Tab:table_parameters}, other fixed parameters are mentioned in the beginning of this section and we take $(\beta_2, \beta_3, \beta_4, \beta_5)=(0.12, 0.1, 0.08, 0.06)$, $( \alpha_2, \alpha_3, \alpha_4, \alpha_5)=( 0.35, 0.3, 0.25, 0.2)$, $0< \beta_1 < 0.5$, $0< \alpha_1 < 1$ and $0< \epsilon < 0.5$.}
    \label{fig:countour_R0}
\end{figure}

Further, we examine contour plots of total infected people in all the five patches ($\sum_{j=1}^{5} I_{i}(t)$) to clarify the effects of the parameters $\epsilon$, $\alpha_1$ and $\beta_1$. From Fig. \ref{fig:countour_infection}, it can be observed that the total infection has nonlinear relationships with the transmission parameters $\alpha_1$ and $\beta_1$ and the coupling strength $\epsilon$. However, for high values of coupling strength ($\epsilon > 0.25$), both the transmission coefficients have positive impact on the total infection. These contour plots reveal that the coupling strength is very crucial in disease spread and the impact may be counter-intuitive depending on situations. 

\begin{figure}
    \centering
    \includegraphics[width=0.45\textwidth]{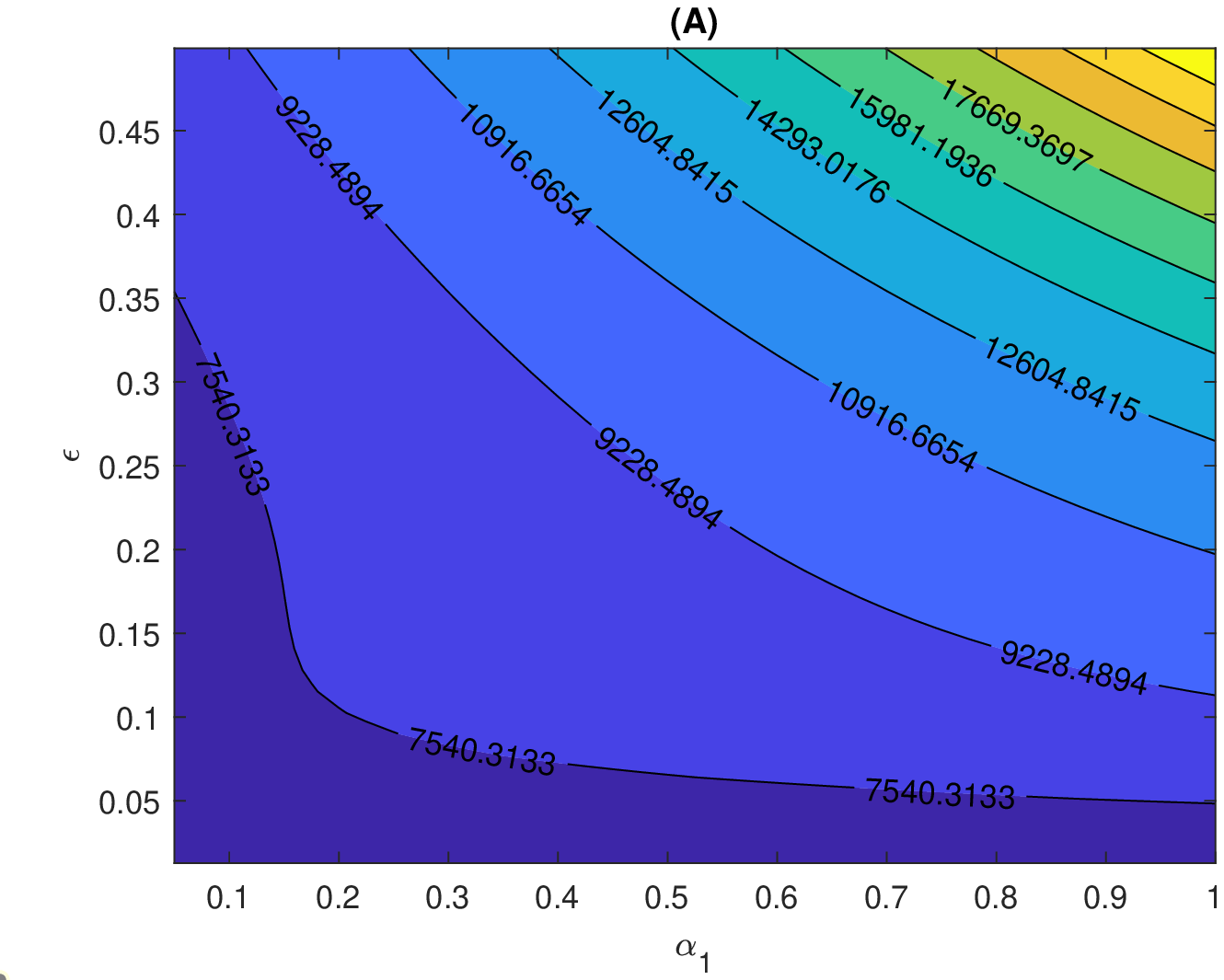}
    \includegraphics[width=0.45\textwidth]{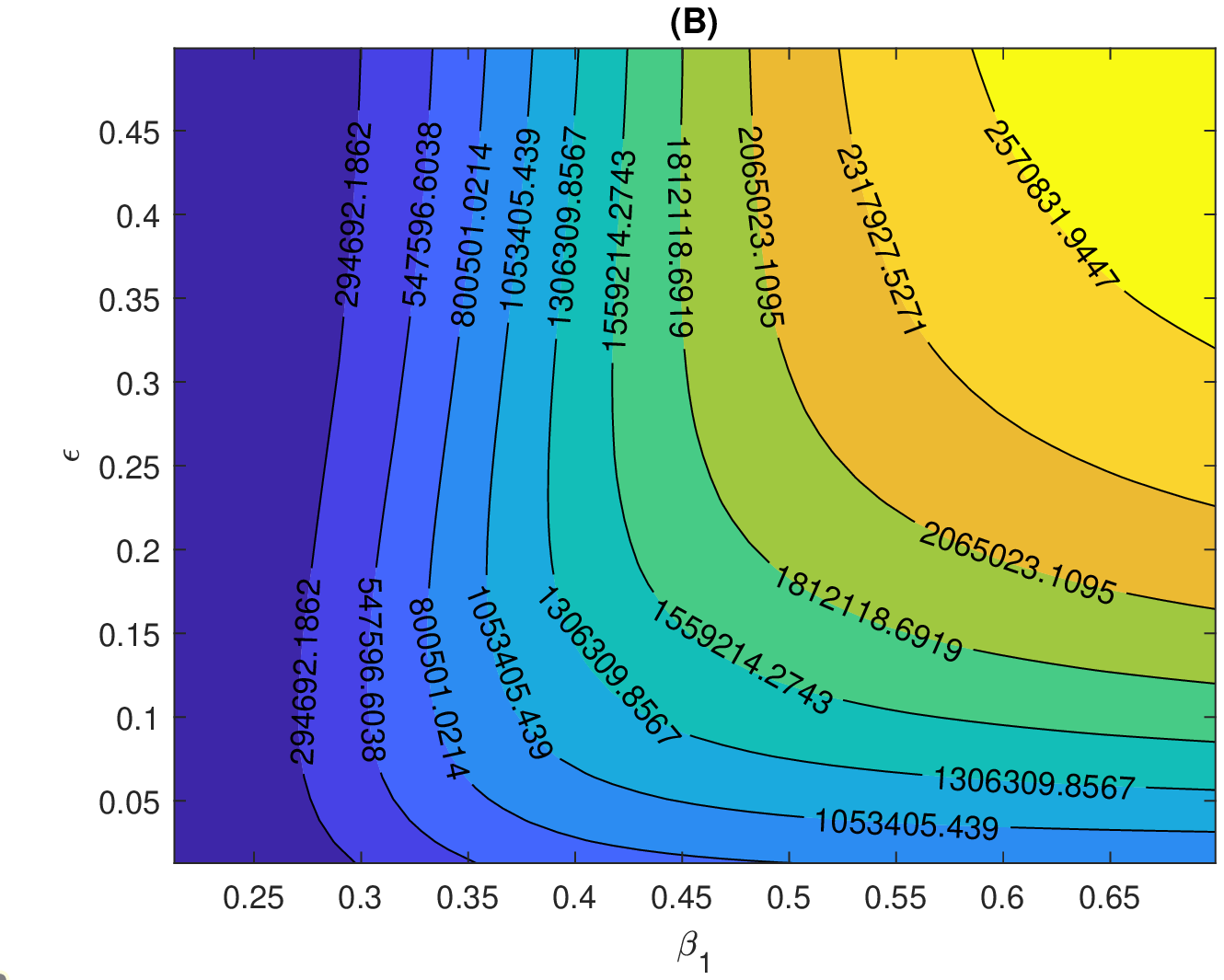}
    \caption{Contour plots of the total infectious individuals number with respect to (A) $\epsilon$ and $\alpha_1$ and (B) $\epsilon$ and $\beta_1$. The fixed parameters are taken from Table \ref{Tab:table_parameters}, other fixed parameters are mentioned in the beginning of this section and we take $(\beta_2, \beta_3, \beta_4, \beta_5)=(0.12, 0.1, 0.08, 0.06)$, $( \alpha_2, \alpha_3, \alpha_4, \alpha_5)=( 0.35, 0.3, 0.25, 0.2)$, $0< \beta_1 < 0.7$, $0< \alpha_1 < 1$ and $0< \epsilon < 0.5$.}
    \label{fig:countour_infection}
\end{figure}

\subsection{Exit screening scenarios}
Long distance passengers are checked for symptoms of an ongoing outbreak in a particular region. This checking or screening are necessary for respiratory disease outbreaks. Screening can be done before departure (exit screening) or after arrival (entry screening) to a certain place \cite{arino2016revisiting}. For instance, WHO requested all SARS-CoV affected areas to screen departing passengers for SARS-CoV symptoms from March to May 2003 \cite{john2005border}. Many Asian countries restarted these screenings during the 2009 pH1N1 spread. More recently, due to the high infectiousness of SARS-CoV-2 most of the countries have implemented exit screening as well as entry screening \cite{johansson2021reducing}. Thermal scanners and test-kits have been used to screen the passengers for any symptoms. Here we consider imperfect border screening is implemented in all the patches with a certain level of efficiency. We assume that infectious patients are screened and restricted from travelling while other people are allowed to travel. We introduce a parameter $\kappa$ which measures the efficacy of exit screening. After incorporation of imperfect exit screening at a rate $\kappa$, the model \eqref{EQ:eqn 2.1} become

\begin{eqnarray}\label{EQ:eqn 7.1}
\displaystyle{\frac{dS_i}{dt}} &=& \Pi_i- \beta_i \frac{I_i}{N_i} S_i - \mu_i S_i - \sum_{j=1}^{n} m_{ji} S_i + \sum_{j=1}^{n} m_{ij} (1-\frac{\alpha_j I_j}{N_j}) S_j,\nonumber \\
\displaystyle{\frac{dE_i}{dt}} &=& \beta_i \frac{I_i}{N_i} S_i -(\gamma_i+\mu_i)E_i - \sum_{j=1}^{n} m_{ji} E_i + \sum_{j=1}^{n} m_{ij} \frac{\alpha_j I_j}{N_j} S_j + \sum_{j=1}^{n} (1-\xi_i)m_{ij} E_j, \\
\displaystyle{\frac{dI_i}{dt}} &=& \gamma_i E_i - (\sigma_i+ \mu_i +\delta_i)I_i - \kappa \sum_{j=1}^{n} m_{ji} I_i + \sum_{j=1}^{n} \xi_i m_{ij} E_j + \kappa \sum_{j=1}^{n} (1-p_i) m_{ij} I_j, \nonumber \\
\displaystyle{\frac{dR_i}{dt}} &=& \sigma_i I_i - \mu_i  R_i - \sum_{j=1}^{n} m_{ji} R_i + \sum_{j=1}^{n} m_{ij} R_j + \kappa \sum_{j=1}^{n} p_i m_{ij} I_j, \nonumber
\end{eqnarray}

The time evolution of the infected compartments of different patches are depicted in Fig. \ref{fig:exit_screening}. The screening efficacy $\kappa$ is taken to be $\kappa =1$ for no screening and decreasing values of $\kappa$ indicate increasing screening efficacy. All the panels in Fig. \ref{fig:exit_screening} show similar trends of decrease in infection prevalence. However, since the infection prevalence is higher in patch 1, exit screening will have negligible effect in the prevalence. To quantify the effects of different exit screening levels, we now calculate the percentage reduction in the infection during 1000 days from infection onset. The percentage reduction is computed using the following basic formula

\begin{eqnarray*}
	\textmd{Percentage reduction in infected persons}=\frac{\textmd{Base value of infected persons} -
		\textmd{Model output}}{\textmd{Base value of infected persons}}\times 100.
\end{eqnarray*}

\begin{figure}
    \centering
    \includegraphics[width=0.45\textwidth]{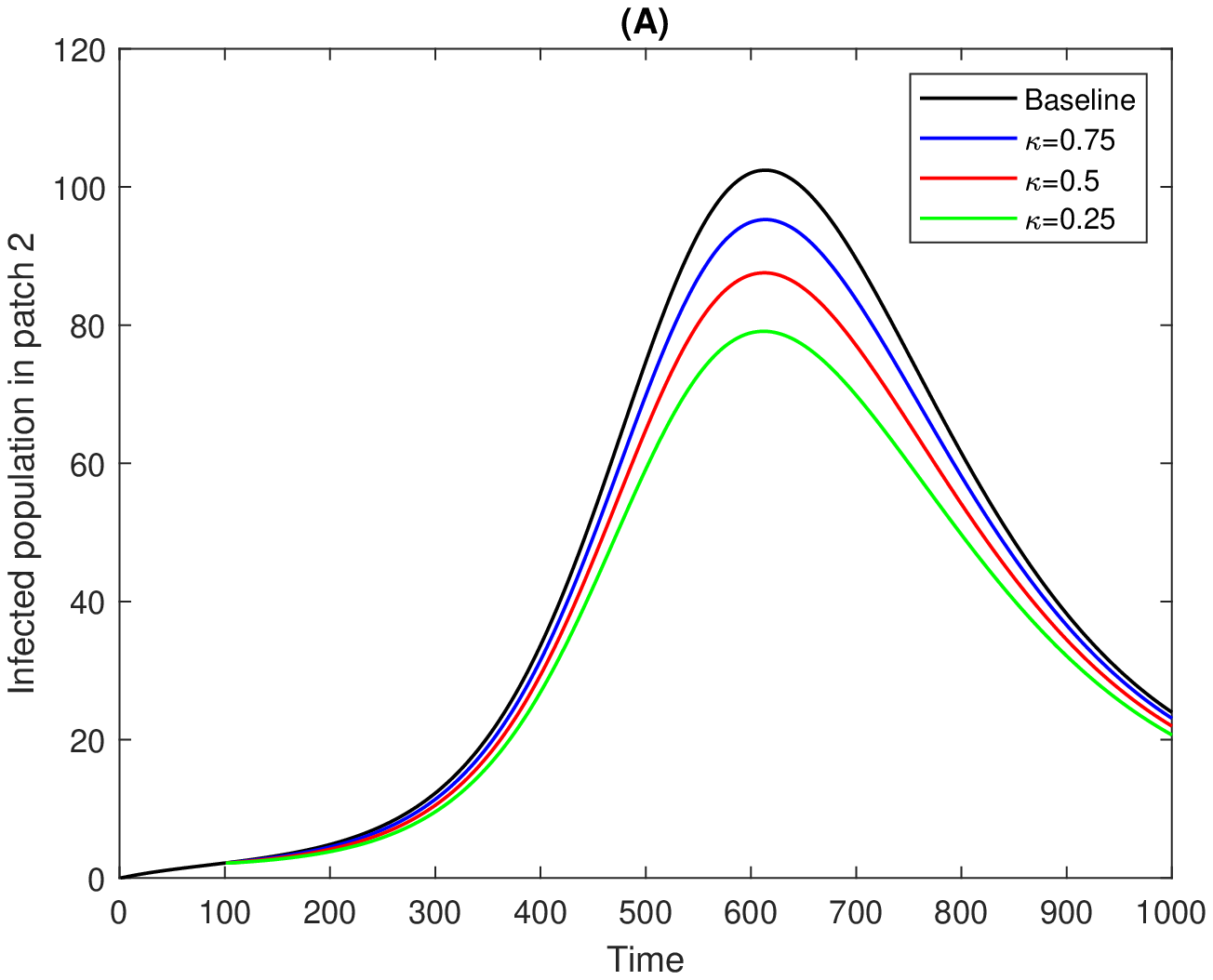}
    \includegraphics[width=0.45\textwidth]{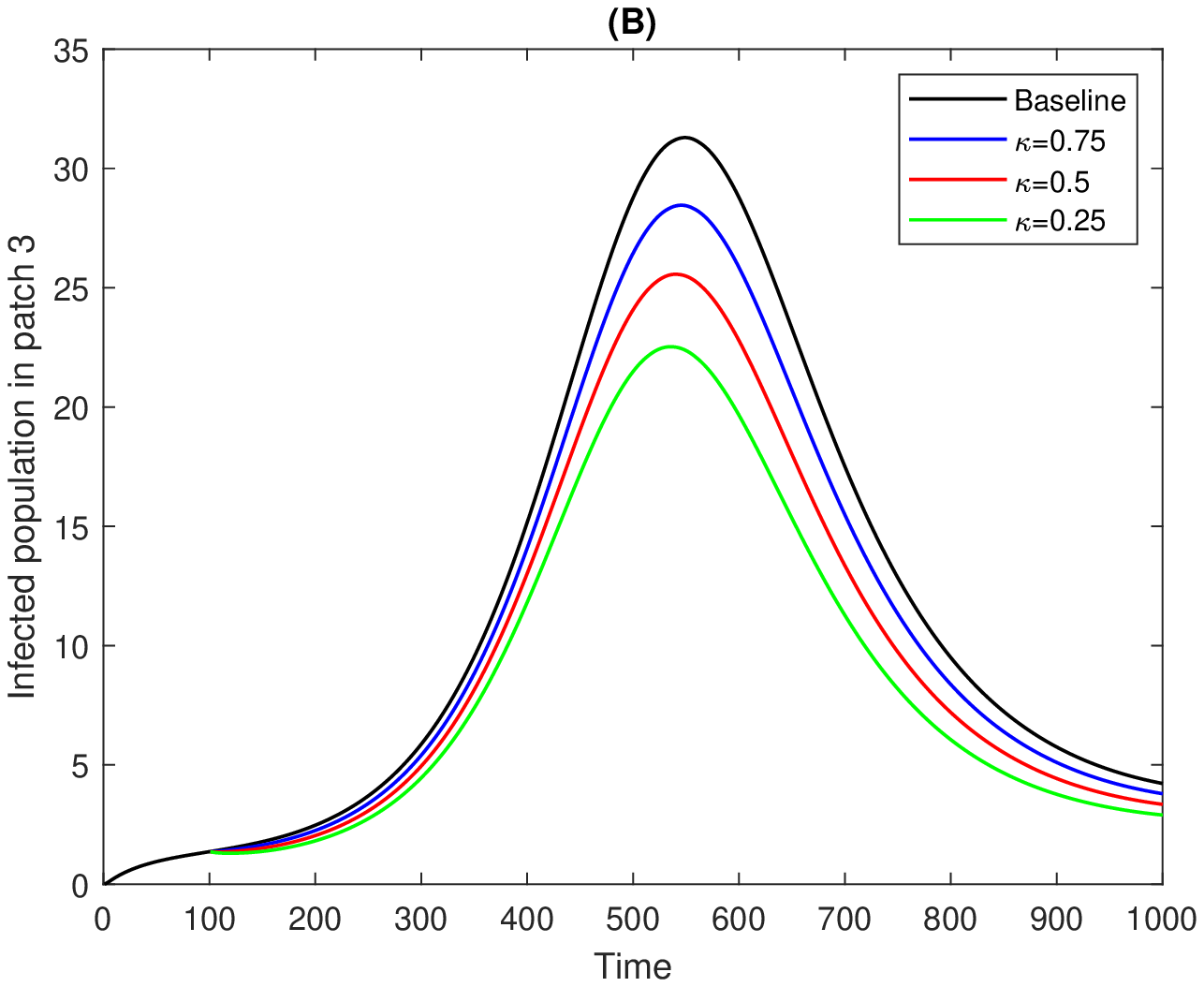}
    \includegraphics[width=0.45\textwidth]{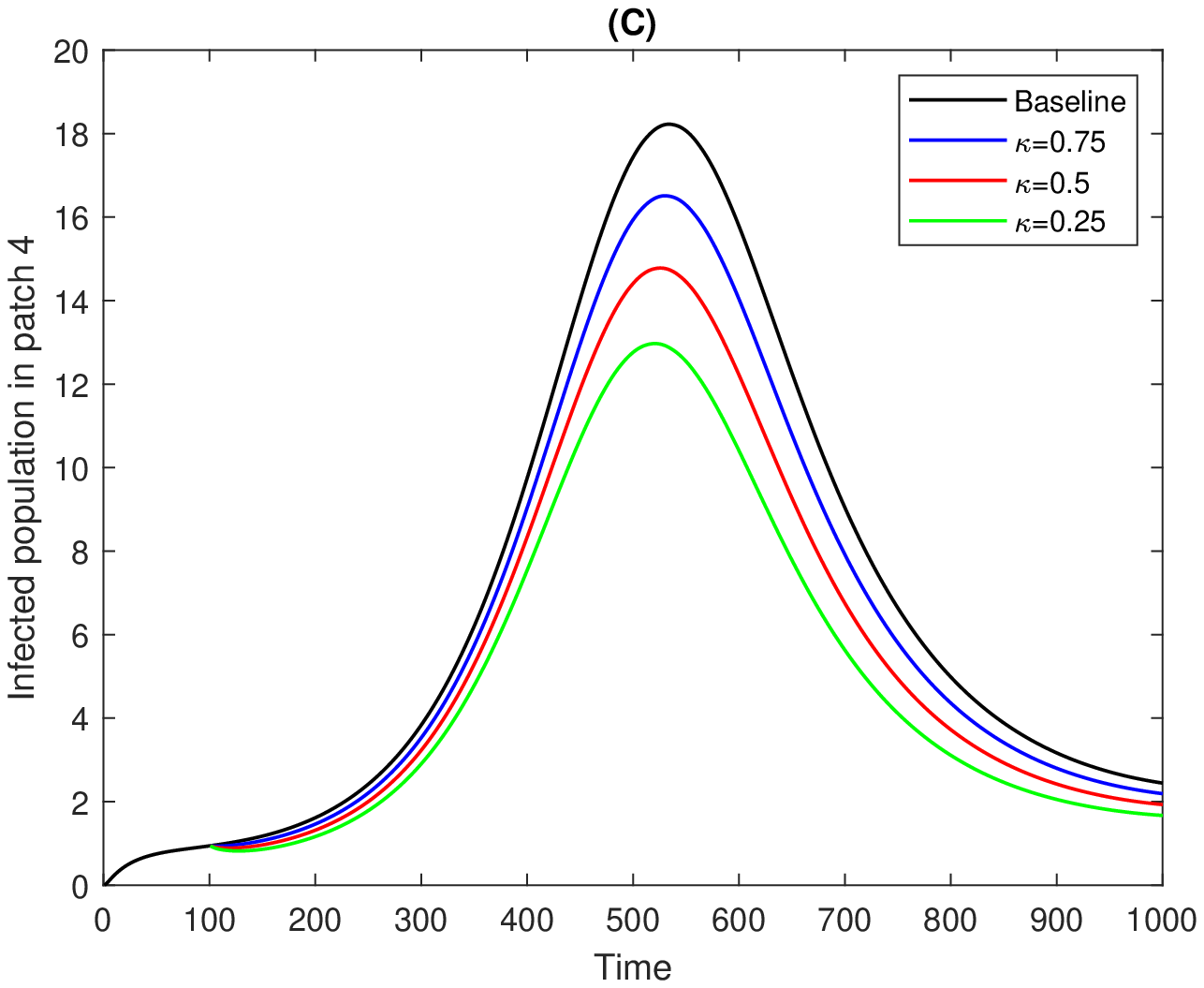}
    \includegraphics[width=0.45\textwidth]{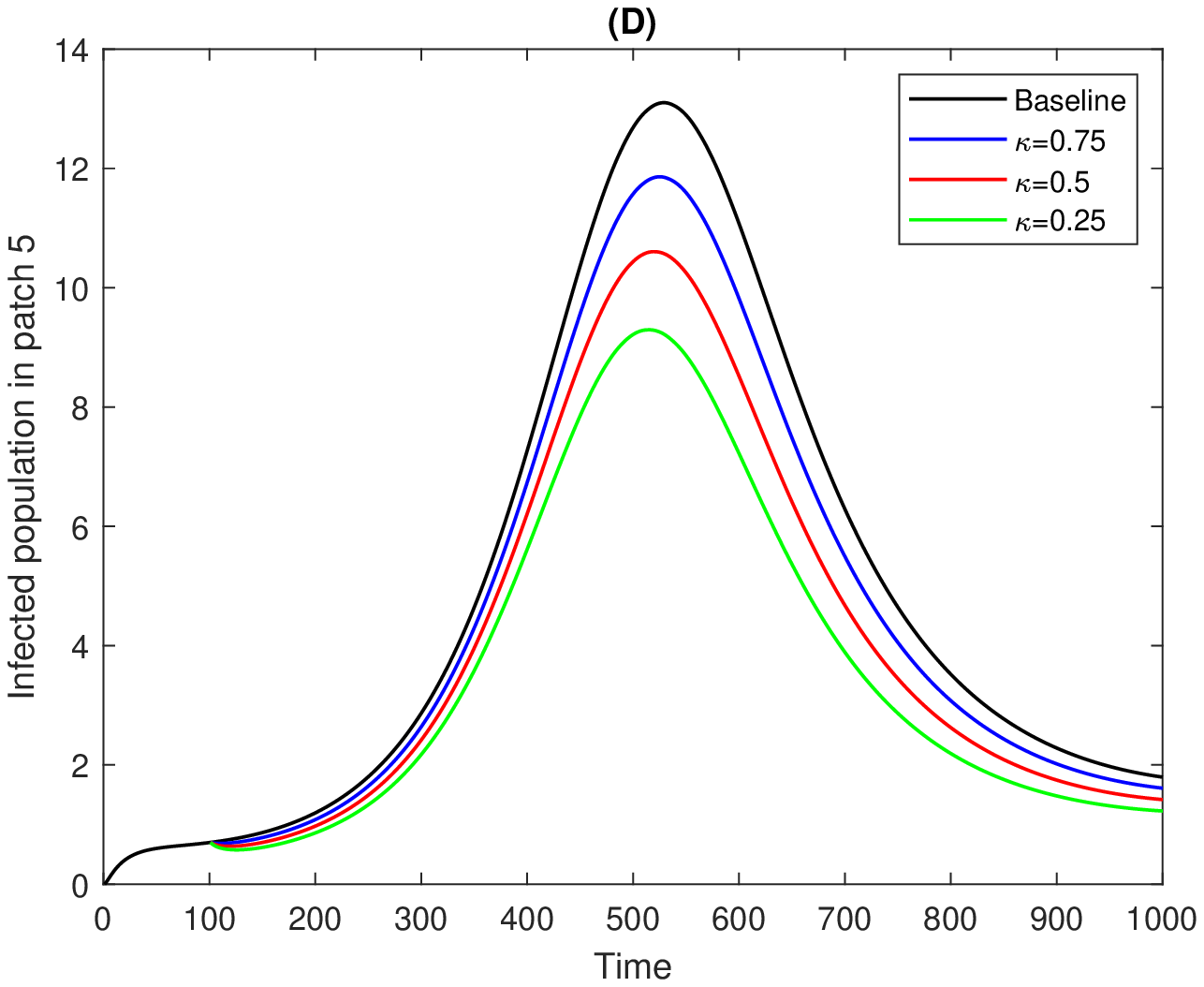}
    \caption{Effects of different imperfect exit screening levels on infectious people. The fixed parameters are taken from Table \ref{Tab:table_parameters}, other fixed parameters are mentioned in the beginning of this section and we take $(\beta_1, \beta_2, \beta_3, \beta_4, \beta_5)=(0.14, 0.12, 0.1, 0.08, 0.06)$, $(\alpha_1, \alpha_2, \alpha_3, \alpha_4, \alpha_5)=(0.4, 0.35, 0.3, 0.25, 0.2)$ and $\epsilon = 0.01$.}
    \label{fig:exit_screening}
\end{figure}

The percentage reductions are reported in Table \ref{Tab:percent-reduction-Table}.
It can be observed that the infection in the patch 1 will increase for different exit screening scenarios. However, all the other patches show significant decrease in cases. This indicate that in a metapopulation, the patch with high prevalence and high transmission rate will not benefit from the exit screening. But, the collectively it can be inferred that the infection will show decreasing trends if exit screening is implemented. 

\begin{table}[h]	
	\centering
	\caption{\bf{Percentage reduction in infectious people for different exit screening scenarios.}}\vspace{0.3cm}
	\begin{tabular}{ccccccc} \hline
		\textbf{Parameter} & \textbf{Values} & \textbf{$I_1(t)$} & \textbf{$I_2(t)$} & \textbf{$I_3(t)$} & \textbf{$I_4(t)$} & \textbf{$I_5(t)$} \\ \hline
		$\kappa$ & 0.75 & -1.13 & 6.22 & 9.50 & 9.55 & 10.05 \\ 
		 & 0.5 & -2.30 & 13.03 & 19.27 & 20.09 & 20.29 \\  
		 & 0.25 & -3.45 & 20.59 & 29.37 & 30.50 & 30.77 \\
		\hline
	\end{tabular}
	\label{Tab:percent-reduction-Table}
\end{table}

\subsection{Emergence of a new strain}
In this subsection, we numerically study the effect of a new strain of the virus on the overall population dynamics. Respiratory disease causing viruses such as influenza, SARS-CoV, MERS-CoV, SARS-CoV-2 etc are mostly RNA viruses and are prone to mutations \cite{lyons2018mutation,garcia2021multiple}. Therefore, it is important to investigate multi-strain dynamics of the metapopulation. We consider a simple two strain model for the propagation of the virus (see Fig. \ref{fig:flow_diagram_two_strain}) in each patch \cite{khyar2020global}.  

\begin{figure}[h]
    \centering
    \includegraphics[width=0.75 \textwidth]{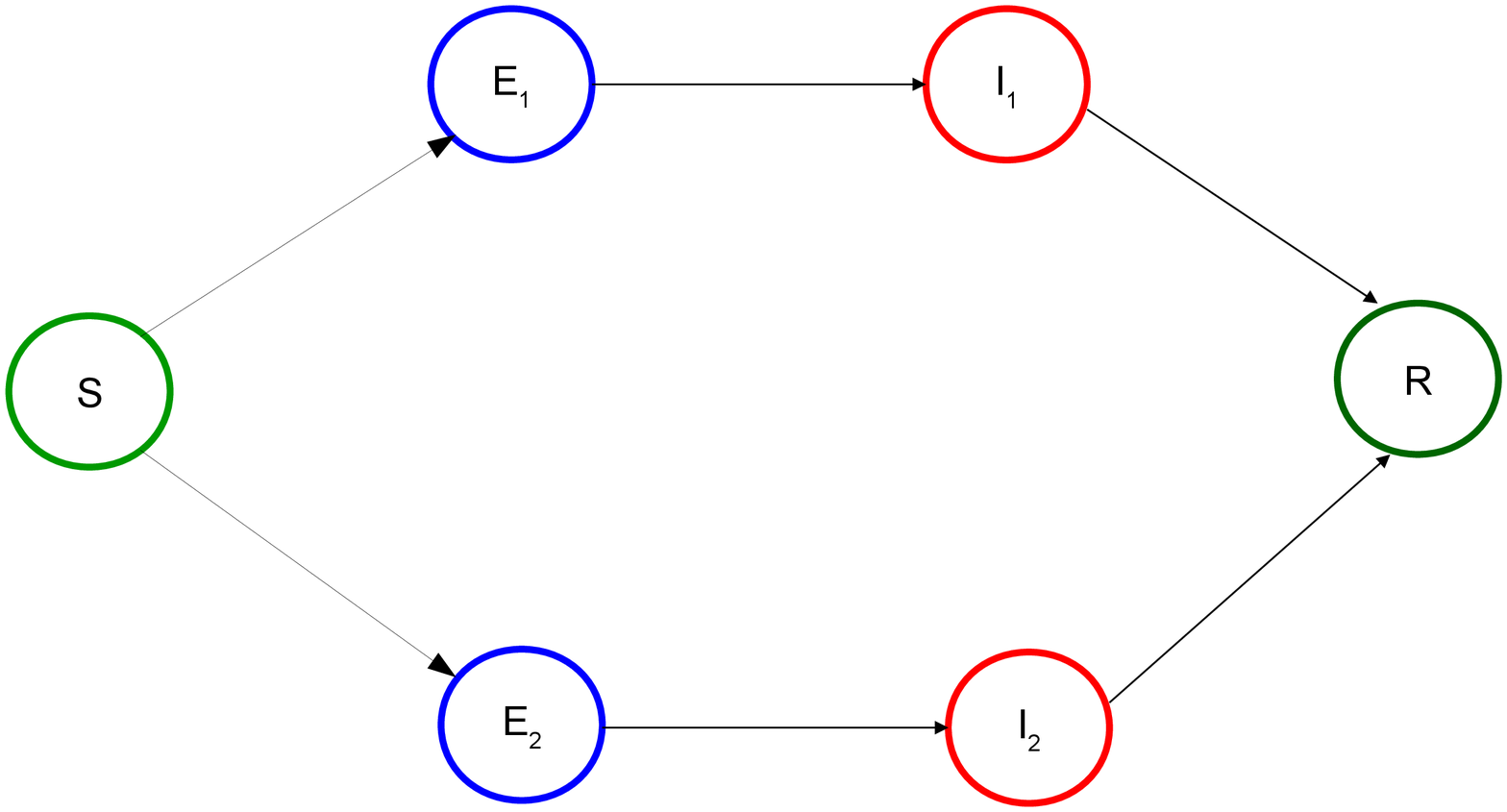}
    \caption{Flow diagram of an SEIR model with two strains co-circulating in a single patch. Solid black arrows depict usual progressions of a compartmental SEIR model.}
    \label{fig:flow_diagram_two_strain}
\end{figure}

We assume that the patches are connected through migration and the transmission of both the strains are possible during travel. Further, the transmission rates for both the strains of the virus are taken to be different whereas all other parameters are taken to be same. Thus, taking these assumptions into consideration, the system \eqref{EQ:eqn 2.1} becomes the following system of equations

\begin{eqnarray}\label{EQ:eqn 8.1}
\displaystyle{\frac{dS_i}{dt}} &=& \Pi_i-  \frac{\beta^1_i I^1_i + \beta^2_i I^2_i}{N_i} S_i - \mu_i S_i - \sum_{j=1}^{n} m_{ji} S_i + \sum_{j=1}^{n} m_{ij} (1-\frac{\alpha^1_j I^1_j + \alpha^2_j I^2_j}{N_j}) S_j,\nonumber \\
\displaystyle{\frac{dE^1_i}{dt}} &=& \beta^1_i \frac{I^1_i}{N_i} S_i -(\gamma_i+\mu_i)E^1_i - \sum_{j=1}^{n} m_{ji} E^1_i + \sum_{j=1}^{n} m_{ij} \frac{\alpha^1_j I^1_j}{N_j} S_j + \sum_{j=1}^{n} (1-\xi_i)m_{ij} E^1_j, \nonumber \\ 
\displaystyle{\frac{dE^2_i}{dt}} &=& \beta^2_i \frac{I^2_i}{N_i} S_i -(\gamma_i+\mu_i)E^2_i - \sum_{j=1}^{n} m_{ji} E^2_i + \sum_{j=1}^{n} m_{ij} \frac{\alpha^2_j I^2_j}{N_j} S_j + \sum_{j=1}^{n} (1-\xi_i)m_{ij} E^2_j, \\
\displaystyle{\frac{dI^1_i}{dt}} &=& \gamma_i E^1_i - (\sigma_i+ \mu_i +\delta_i)I^1_i - \sum_{j=1}^{n} m_{ji} I^1_i + \sum_{j=1}^{n} \xi_i m_{ij} E^1_j + \sum_{j=1}^{n} (1-p_i) m_{ij} I^1_j, \nonumber \\
\displaystyle{\frac{dI^2_i}{dt}} &=& \gamma_i E^2_i - (\sigma_i+ \mu_i +\delta_i)I^2_i - \sum_{j=1}^{n} m_{ji} I^2_i + \sum_{j=1}^{n} \xi_i m_{ij} E^2_j + \sum_{j=1}^{n} (1-p_i) m_{ij} I^2_j, \nonumber \\
\displaystyle{\frac{dR_i}{dt}} &=& \sigma_i (I^1_i + I^2_i) - \mu_i  R_i - \sum_{j=1}^{n} m_{ji} R_i + \sum_{j=1}^{n} m_{ij} R_j + \sum_{j=1}^{n} p_i m_{ij} (I^1_j + I^2_j), \nonumber
\end{eqnarray}

The model \eqref{EQ:eqn 8.1} is simulated using parameter values from Table \ref{Tab:table_parameters} and other parameters are taken as $\epsilon = 0.01$ 
$(\alpha^1_1, \alpha^1_2, \alpha^1_3, \alpha^1_4, \alpha^1_5)=( 0.16, 0.14, 0.12, 0.1,0.08)$, $(\alpha^2_1, \alpha^2_2, \alpha^2_3, \alpha^2_4, \alpha^2_5)=( 0.2, 0.2, 0.2, 0.2,0.2)$, $d_{12}=d_{21}=100$, $d_{13}=d_{31}=110$, $d_{14}=d_{41}=120$, $d_{15}=d_{51}=130$, $d_{23}=d_{32}=140$, $d_{24}=d_{42}=150$, $d_{25}=d_{52}=160$, $d_{34}=d_{43}=170$, $d_{35}=d_{53}=180$, $d_{45}=d_{54}=190$ and $(\nu_1, \nu_2, \nu_3, \nu_4, \nu_5)=(7, 6, 5, 4, 3)$. $\beta^1_i$ and $\beta^2_i$ are varied in the range $(0.1,0.4)$. Initially, we examine the time evolution of both the strains. To this end, we simulate the model \eqref{EQ:eqn 2.1} with different transmission rates and initial conditions $(S_1(0), E_1(0), I_1(0), R_1(0))=(100000, 100, 10, 0)$ and $(S_i(0), E_i(0), I_i(0), R_i(0))=(100000, 0, 0, 0)$ for i=2,3,4,5. After 200 days, a second strain is seeded in patch 3. To do this numerically, we take the end points of the state variables from the 200$^{th}$ day and plug in them to the system \eqref{EQ:eqn 8.1} along with the number of strain 2 infected in patch 3 is 1. However, the initial strain 2 infected persons ($I^2_i$, i=2,3,4,5) in other patches and strain 2 exposed persons ($E^2_i$, i=1,2,3,4,5) in all patches are taken to be zero. Depending on the transmission coefficients $\beta^1_i$ and $\beta^2_i$, we observe that either the first strain infected persons persist (see Fig. \ref{fig:first_strain_persist}) or the second strain infected persons persist (see Fig. \ref{fig:second_strain_persist}).

\begin{figure}[h]
    \centering
    \includegraphics[width=0.45\textwidth]{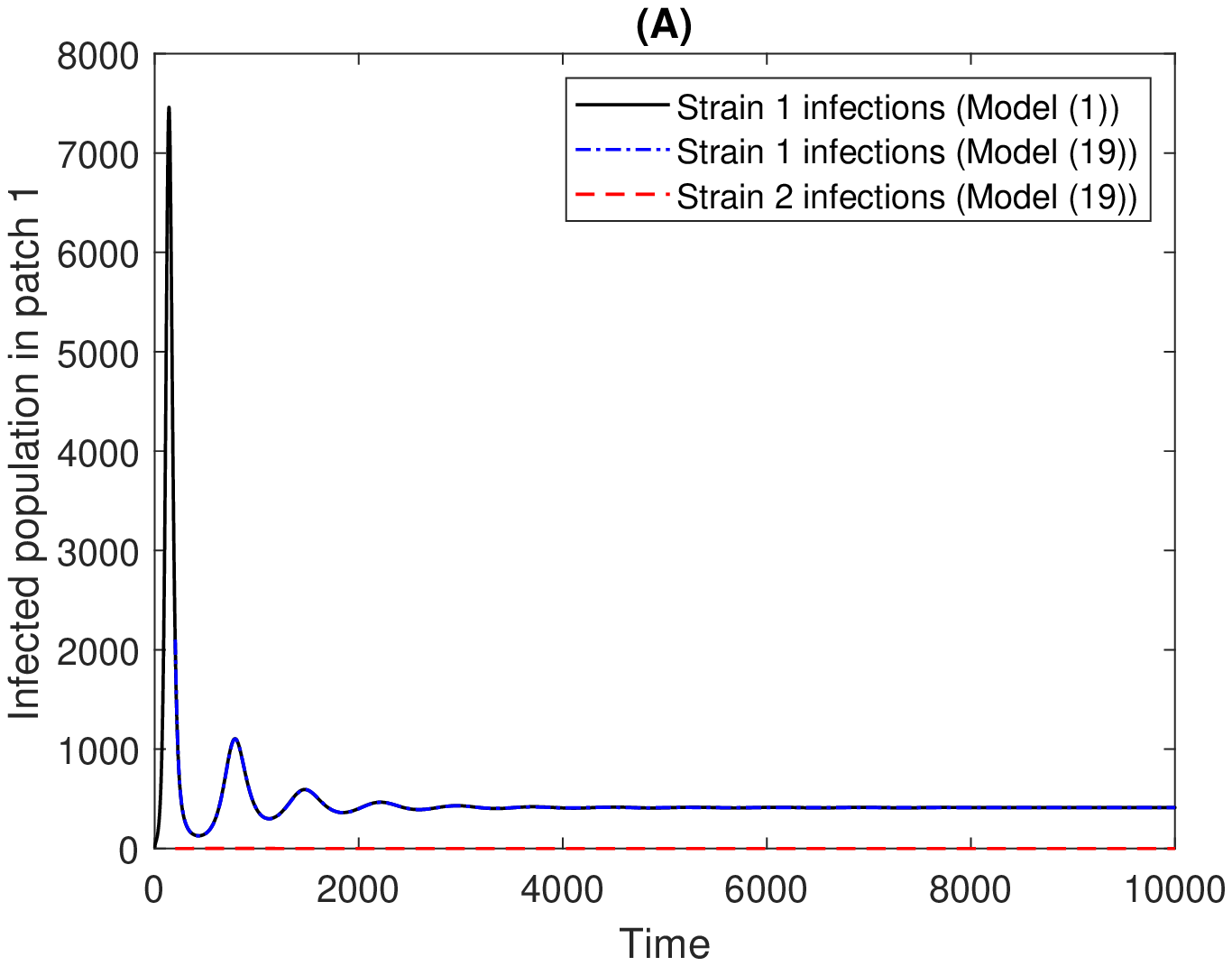}
    \includegraphics[width=0.45\textwidth]{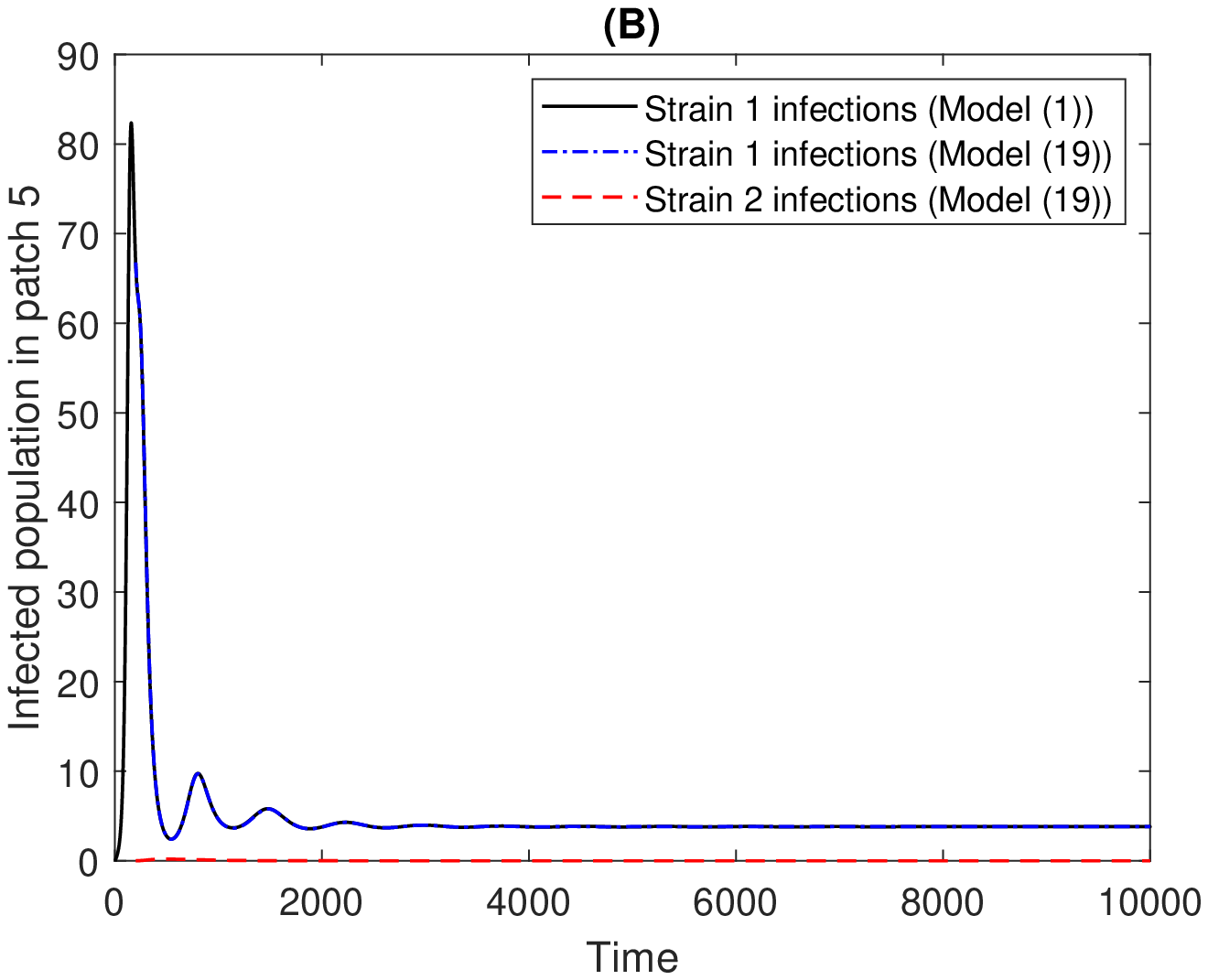}
    \caption{Persistence of first strain in (A) patch 1 and (B) patch 5. The transmission coefficients are $(\beta^1_1, \beta^1_2, \beta^1_3, \beta^1_4, \beta^1_5)=(0.24, 0.2, 0.16, 0.12,0.08)$ and $(\beta^2_1, \beta^2_2, \beta^2_3, \beta^2_4, \beta^2_5)=(0.2, 0.2, 0.2, 0.2,0.2)$.}
    \label{fig:first_strain_persist}
\end{figure}

\begin{figure}[h]
    \centering
    \includegraphics[width=0.45\textwidth]{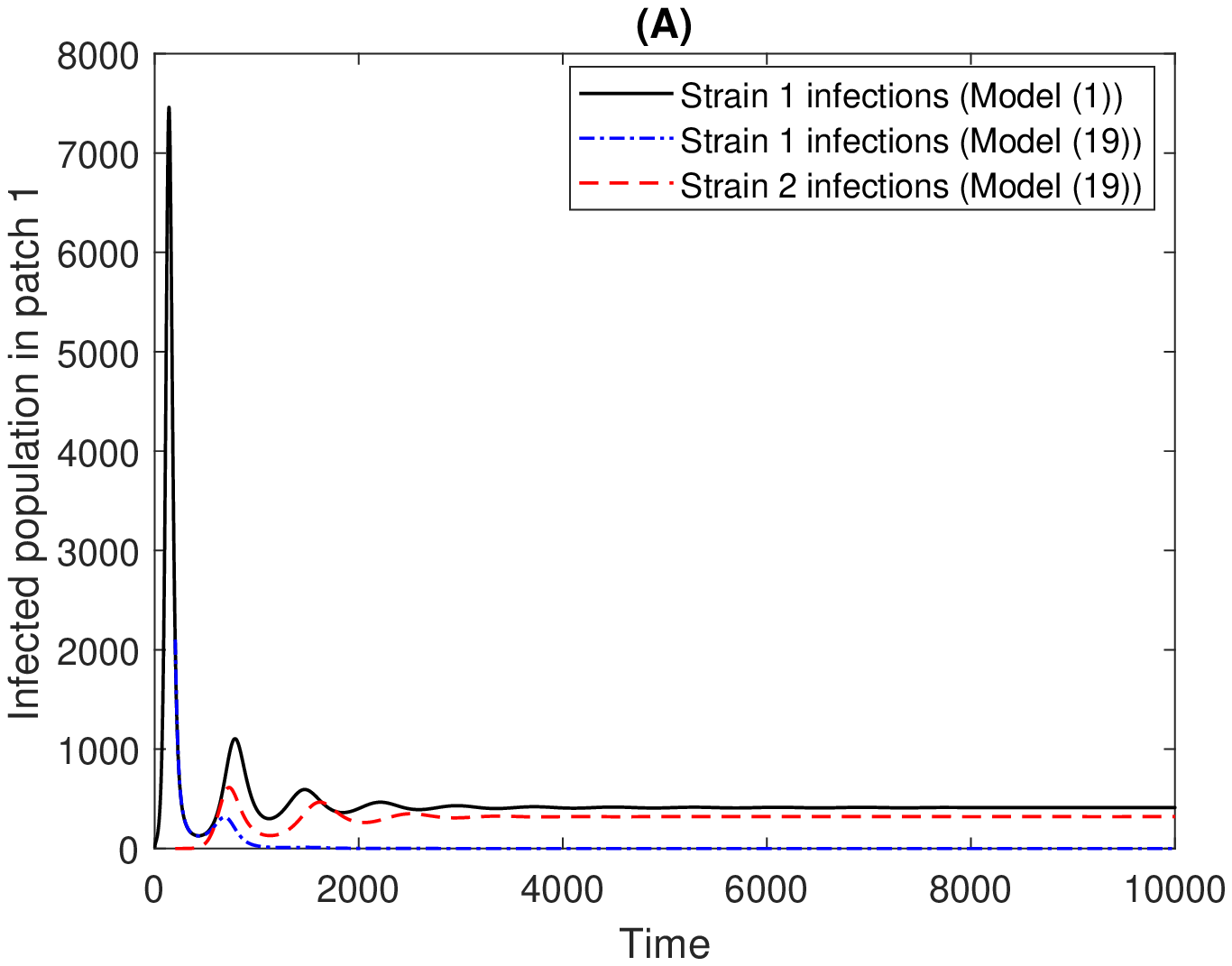}
    \includegraphics[width=0.45\textwidth]{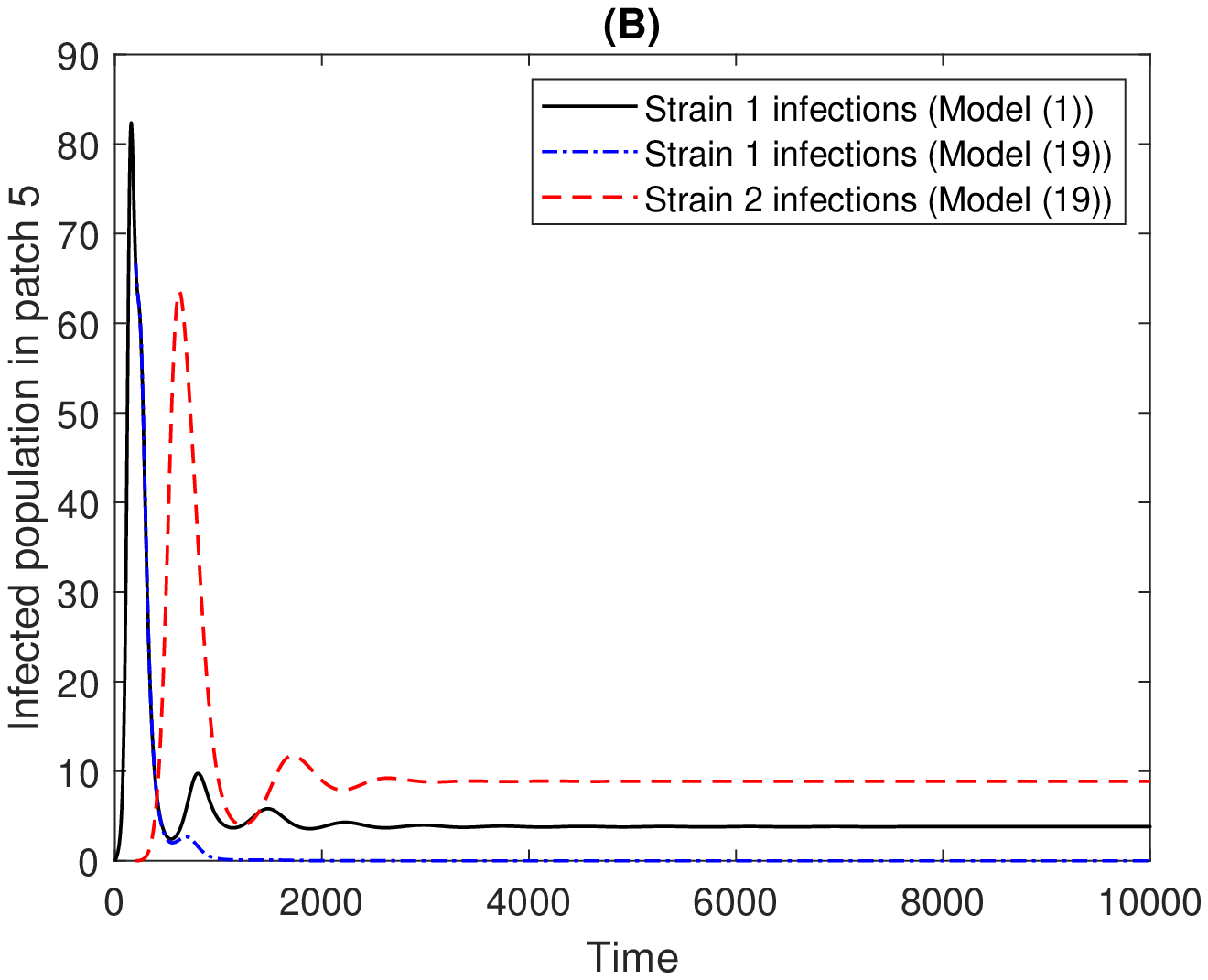}
    \caption{Persistence of strain 2 infection and extinction of strain 1 infections in (A) patch 1 and (B) patch 5. The transmission coefficients are $(\beta^1_1, \beta^1_2, \beta^1_3, \beta^1_4, \beta^1_5)=(0.24, 0.2, 0.16, 0.12,0.08)$ and $(\beta^2_1, \beta^2_2, \beta^2_3, \beta^2_4, \beta^2_5)=(0.25, 0.25, 0.25, 0.25,0.25)$.}
    \label{fig:second_strain_persist}
\end{figure}

It can be observed that the two co-circulating strains will not co-exist in the community in a long run. This effect is know as competitive exclusion principle in mathematical epidemiology \cite{martcheva2015introduction}. Fig. \ref{fig:second_strain_persist} depicts the persistence of second strain infection and extinction of strain 1 infection in patch 1 and patch 5. It can be observed that both the strains co-exist for a very short period. However, a transition is observed for both the patches i.e, infected persons with strain 1 
go to extinction after emergence of the strain 2 infection. Additionally, it is seen that strain 1 infection may have persisted if second strain was not introduced in the population. 

\begin{figure}[h]
    \centering
    \includegraphics[width=0.45\textwidth]{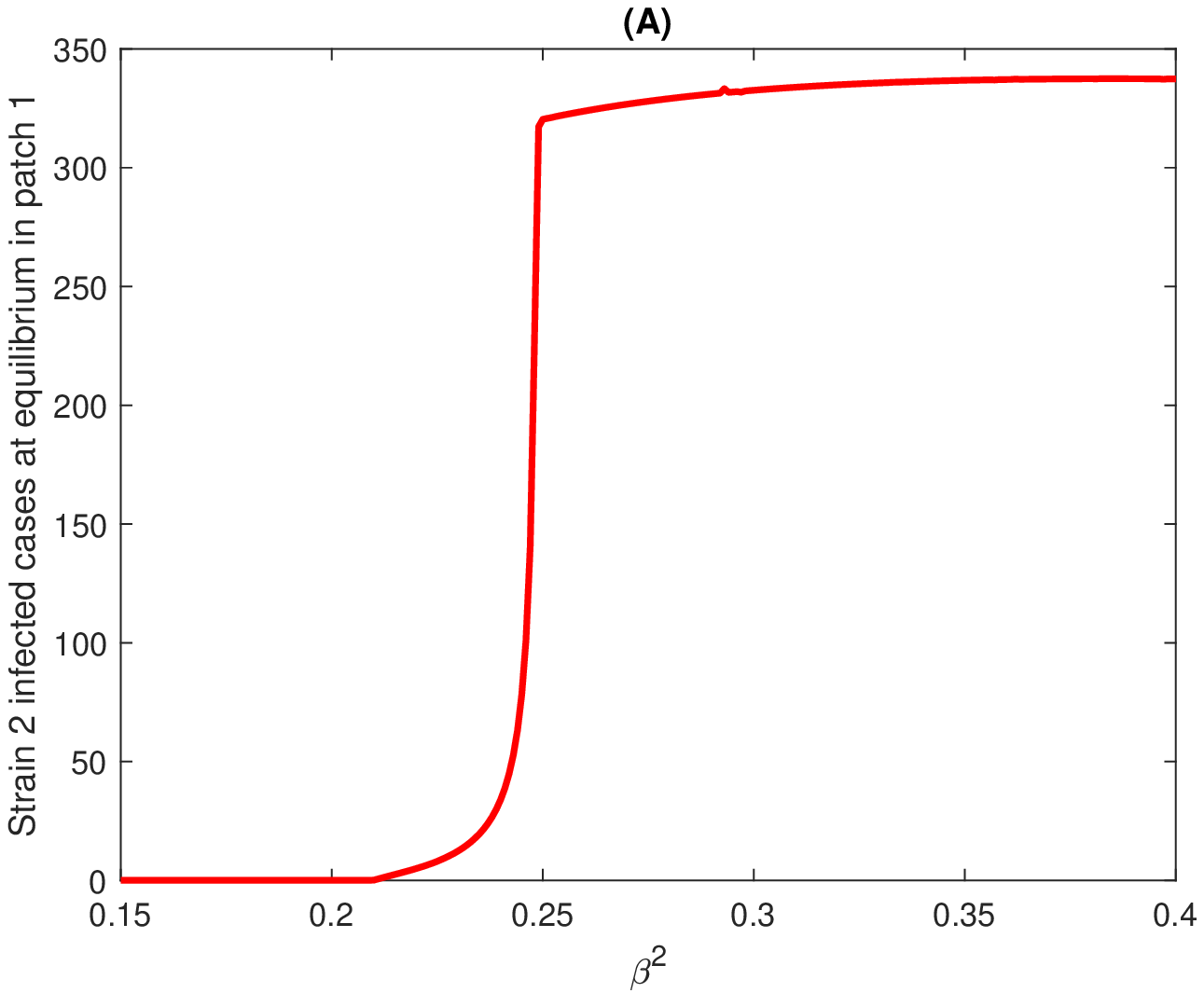}
    \includegraphics[width=0.45\textwidth]{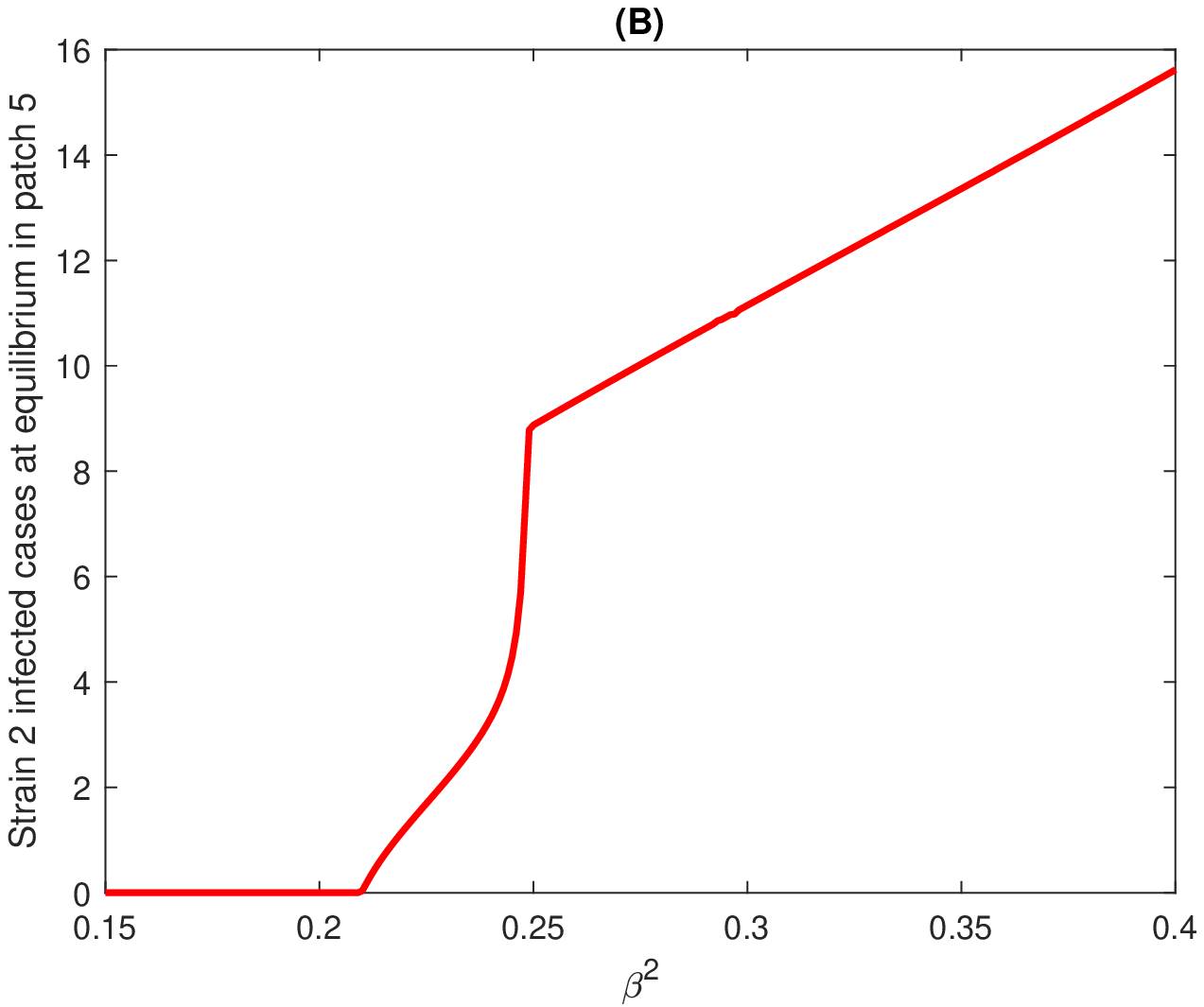}
    \caption{Dynamics of the strain 2 equilibrium density with different values of $\beta^2_i$ in (A) patch 1 and (B) patch 5.}
    \label{fig:threshold_beta2}
\end{figure}

Further, to investigate the behavior of the equilibrium value of the second strain with respect to the transmission coefficient $\beta^2_i$, we draw the bifurcation diagrams \ref{fig:threshold_beta2}. It can be observed that the equilibrium value remains at zero until $\beta^2_i$ cross a certain threshold. This gives rise to a forward transcritical bifurcation. This indicate that the strain 2 free equilibrium and strain 2 endemic equilibrium will change stability after a certain threshold value of $\beta^2_i$.

\section{Discussion}\label{section5}
In this paper, we propose and analyze a novel metapopulation model with infection during transport. We extend a two-patch SIS model of Arino et. al \cite{arino2016revisiting} to incorporate recovery and incubation during transport. We also consider the migration terms to be related to distance between patches and perceived severity of the disease. The general n-patch model is analyzed mathematically. Positivity and boundedness of the proposed model are established and an invariant region for the system is obtained. Existence of the unique disease-free equilibrium has been proved and the local stability of this equilibrium is governed by the basic reproduction number ($R_0$). Moreover, if the movement matrices satisfy certain conditions and $R_0 < 1$, then the DFE is globally asymptotically stable. Existence of an endemic equilibrium is also established under some conditions.

Extensive numerical experiments using different parameter sets are performed to get insight into the transmission process. Different migration matrices are considered to study the network topology such as fully connected, ring of patches or star-like network (as depicted in Fig. \ref{fig:patch_structures}). Keeping all the parameters fixed except the migration matrix, we simulate the system which reveal that different network topology does have a effect on the prevalence of the disease in different patches (see Fig. \ref{fig:patch_structures_ts}). Further numerical simulations suggest that infection during travel has the potential to move the equilibrium from disease free to an endemic state (see Fig. \ref{fig:with_without_IDT}). The transmission rates during travel ($\alpha_i$) are then varied individually to quantify their effects on disease prevalence. It is observed that all the $\alpha_i$'s play significant role in disease transmission (see Fig. \ref{fig:alpha_box_plots}). Furthermore, numerically we observe that depending on the value of coupling strength, transmission coefficients ($\beta_1$ and $\alpha_1$) may show nonlinear effects on $R_0$ and total number of infectives. The coupling strength also show some nonlinear relationships with both $R_0$ and total number of infectious persons (as seen in Fig. \ref{fig:countour_R0} and Fig. \ref{fig:countour_infection}). Epidemiologically, this indicate that while dealing real epidemic curves, modellers have to be very careful about the choice of coupling strengths. Modelling diseases in a metapopulation setting is inherently related to the migration rates between patches. From control strategic point of view, restricting migration of infectives between patches are studied for disease suppression. Thus, various exit screening scenarios were studied to get a better understanding of the system. We observe that implementation of imperfect exit screening in all the patches may have negative impact on the patch with high prevalence and high transmission rate (patch 1 in this case). However, all other patches experience significant reduction in infectious persons.  Finally, we studied the emergence of a new variant strain of the virus and its effect on the dynamics of the previous strain. It is observed that first strain is persistent or the second strain is persistent in the long run depending on the transmission coefficients. However, two strain will not co-exist in the population at equilibrium. The equilibrium values of the second strain also show threshold-like behaviour with respect to the transmission coefficients of the second strain. This indicate that transmission rates are very crucial for the persistence or extinction of the new strain in the community.

\section*{Acknowledgements}
Research of IG is supported by National Board for Higher Mathematics (NBHM) postdoctoral fellowship (Ref. No: 0204/3/2020/R \& D-II/2458). SSN receives senior research fellowship from CSIR, Government of India, New Delhi. The work of SR is partially supported by SERB MATRICS grants MTR/2020/000186 and MSC/2020/00028 of the Government of India.

\bibliographystyle{plain}
\biboptions{square}
\bibliography{bib_metapop}

\end{document}